\documentclass[journal,12pt,draftclsnofoot,onecolumn]{IEEEtran}
\usepackage{epsfig}
\usepackage{booktabs}
\usepackage{epstopdf}
\usepackage{amssymb}
\usepackage{mathtools}
\usepackage{subfigure}

\usepackage{algorithmic}

\usepackage{array}
 \usepackage{mathrsfs}
\usepackage{gensymb}
\usepackage{amssymb}
\usepackage{amsmath}
\usepackage{cite}
\usepackage{url}
\usepackage{xcolor}
\usepackage{cite,graphicx,amsmath,amssymb}
\usepackage{fancyhdr}
\usepackage{mdwmath}
\usepackage{mdwtab}
\usepackage{caption}
\usepackage{amsthm}
\usepackage{suffix}
\usepackage{mathtools}
\allowdisplaybreaks[4]
\usepackage{xpatch}
\usepackage{textcomp}
\usepackage{mathrsfs}
\usepackage{cool}
\usepackage{dsfont}

\xpatchbibmacro{name:andothers}{%
  \bibstring{andothers}%
}{%
  \bibstring[\emph]{andothers}%
}{}{}

\DeclarePairedDelimiterX\MeijerM[3]{\lparen}{\rparen}%
{\begin{smallmatrix}#1 \\ #2\end{smallmatrix}\delimsize\vert\,#3}

\catcode`,\active

\catcode`\,12

\title{Stochastic Analysis of Cooperative Satellite-UAV Communications}

\author{Yu Tian, Gaofeng Pan, $Senior Member, IEEE$, Mustafa A. Kishk, $Member, IEEE$, and Mohamed-Slim Alouini, $Fellow, IEEE$\vspace{-10mm}}

\date{February 2020}

\newtheorem{theorem}{Theorem}

\newtheorem{lemma}{Lemma}

\newtheorem{corollary}{Corollary}

\newtheorem{proposition}{Proposition}

\begin{document}
\maketitle

\begin{abstract}
    In this paper, a dual-hop cooperative satellite-unmanned aerial vehicle (UAV) communication system including a satellite (S), a group of cluster headers (CHs), which are respectively with a group of uniformly distributed UAVs, is considered. Specifically, these CHs serve as aerial decode-and-forward relays to forward the information transmitted by S to UAVs. Moreover, free-space optical (FSO) and radio frequency (RF) technologies are respectively adopted over S-CH and CH-UAV links to exploit FSO's high directivity over long-distance transmission and RF's omnidirectional coverage ability. The positions of the CHs in the 3-dimensional space follow the Matérn hard-core point processes type-II in which each CH can not be closer to any other ones than a predefined distance. Three different cases over CH-UAV links are considered during the performance modeling: interference-free, interference-dominated, and interference-and-noise cases. Then, the coverage performance of S-CH link and the CH-UAV links under three cases is studied and the closed-form analytical expressions of the coverage probability (CP) over both links are derived. Also, the asymptotic expressions for the CP over S-CH link and CH-UAV link in interference-free case are derived. Finally, numerical results are provided to validate our proposed analytical models and thus some meaningful conclusions are achieved.
\end{abstract}

\begin{IEEEkeywords}
Coverage probability, free-space optical communication, Matérn hard-core point process, satellite communication, stochastic geometry, unmanned aerial vehicle
\end{IEEEkeywords}

\section{Introduction}
Exhibiting the advantages of large-scale coverage, abundant frequency resource, and flexibility of deployment, satellite communication has been widely applied in disaster monitoring and rescue, location and navigation, and long-distance information transmission\cite{pan2020performance}. So far, researches mainly focus on satellite-terrestrial communication systems from the aspects of performance analysis \cite{zedini2020performance,pan2020harq,pan2020performance}, resource allocation \cite{pan2020performance,Kawamoto2020flex}, user scheduling \cite{Christopoulos2015muliti}, and physical layer security analysis \cite{zhang2020secrecy,illi2020phy}.

On the other hand, unmanned aerial vehicle (UAV) has been widely used in many applications, including military investigation, disaster relief and rescue, law enforcement, aerial photography, agricultural monitoring, and plant protection, etc., to exploit the virtues of UAV, like small size, light weight, low cost, flexible and fast deployment, and scalability \cite{zolanvari2020potential}. As UAVs can freely change their locations, line-of-sight (LoS) communication links with ground users or stations can be quickly and efficiently built up. Considerable attention has been paid to the resource allocation, trajectory planning, and physical layer security on UAV-ground communications \cite{cai2020joint,lei2020safeguarding,savkin2020securing,pan2020secrecy}. 
However, in the scenarios that ground facilities are destroyed by natural disasters such as earthquakes and floods, or in some inaccessible areas like desert, ocean, and forest, it is hard to build up communication links with UAVs, which results in accidents out of control or even crashing.

To tackle the aforementioned problems, the satellite can serve as an alternative to play a similar role as the terrestrial facilities to set up reliable links between UAVs and the remote control center. Then, by jointly applying the merits of both satellite and UAV communications, satellite-UAV communication systems can introduce more flexibility to the applications of UAVs, especially in harsh application scenarios, e.g., disaster response, natural resource exploration, and military applications in hostile and unfamiliar environments, compared to either traditional satellite communication systems or traditional UAV communication systems. 

Motivated by these observations, there have been plenty of researches presented to design and study satellite-UAV communications in which UAVs work as aerial relays to assist the communications between the satellite and terrestrial terminals \cite{li2020unified,sharma2020outage,huang2020energy, kong2020multiuser,dai2020uav} or aerial terminals\cite{zhou2019beam}. In \cite{li2020unified}, the outage performance of hybrid satellite/UAV terrestrial non-orthogonal multiple access networks in which one UAV served as a relay to forward signals to ground users was investigated and the optimal location of UAV to maximize the sum rate was achieved. In \cite{sharma2020outage}, the outage probability (OP) of a hybrid satellite-terrestrial network in which a group of UAVs are mobile in a three-dimensional (3D) cylindrical space and act as relays was analyzed. In \cite{huang2020energy}, the energy-efficient beamforming was investigated for a satellite-UAV-terrestrial system and in that considered system a multi-antenna UAV works as a relay. In \cite{kong2020multiuser}, the ergodic capacity of an asymmetric free-space optical (FSO)/radio frequency (RF) link in satellite-UAV-terrestrial networks was evaluated while FSO communication was adopted in the satellite-UAV link. In \cite{dai2020uav}, satellite-UAV-ground integrated green Internet of things networks were proposed and studied by optimizing transmit power allocation and UAV trajectory to achieve maximum vehicle rate. In \cite{zhou2019beam}, beam management and self-healing in satellite-UAV mesh millimeter-Wave networks were studied to address the beam misalignment issues between UAVs, and UAV head and satellite/base station. 

Therefore, one can see that UAVs are normally used as aerial relays to serve terrestrial terminals in most of the existing works on satellite-UAV communication systems. To the best of the authors' knowledge, there are no researches presented to investigate the performance of the satellite-UAV systems under the cases that UAVs play as terminals for some specialized application purposes, e.g., photography, observation, surveillance, and strike. When numerous UAVs are deployed, very-high-gain antennas for RF communications or high-accuracy laser receiving systems for FSO communications can not be equipped with UAVs, due to their rigorous hardware limitations. To guarantee the communication quality of the satellite-UAV link in which large path-loss leads to very weak receiving signals, aerial relays with advanced receiving facilities can be deployed to improve the communications between the satellite and UAVs. 


Furthermore, generally, traditional satellite communication links are built up via RF links, and then the transmitted information over such long-distance RF transmission links is quite vulnerable to be wiretapped because of its broadcasting property \cite{pan2020performance}. Benefiting from its inherent metrics, including strong directivity, high security, low risk of eavesdropping, and low interference among different beams, FSO technology has been considered and designed as a promising alternative for the link between the satellite and aerial planes \cite{zedini2020performance,illi2020phy}. 

Motivated by these aforementioned observations, in this work we consider a cooperative satellite-UAV communication system consisting of a satellite (S) and a group of common airplanes/UAVs with superior hardware resources serving as cluster headers (CHs), which are respectively with a group of uniformly distributed UAVs. Specifically, S first transmits information to CHs via FSO links to exploit the high directivity over the long-distance transmission of FSO technology so as to safeguard the information security, and then CHs respectively decode the received information and then forward the recoded information to the UAVs around them by using RF technology. Considering that each CH has its own serving space, the positions of CHs in 3D space follow Matérn hard-core point processes (MHCPP) type-II \cite{haenggi2012stochastic} in which one CH can not be closer to any other CHs than a predefined distance. Moreover, the distribution of UAVs obeys a homogeneous Poisson point processes (HPPP) in the serving space of each CH.

Though there exist considerable works presented to study the performance of traditional terrestrial wireless networks in two-dimensional space by using MHCPP \cite{hunter2008trans,he2016modeling,omri2018distance}, they can not be directly applied to the 3D scenario considered in this work.


The main contributions of this work are summarized as follows.

1) Compared with \cite{li2020air} in which the statistical randomness of the interfering signals in 3D space was approximated by using Gamma distribution via the central limit theorem, in this work the more accurate moment generating function (MGF) of the summation of interfering signals over CH-UAV RF links is derived while considering the randomness of the 3D locations of CHs. 

2) Compared with \cite{li2020air,bach2018model,omri2016model} in which non-closed-form analytical expressions were presented for performance indices while considering 3D interfering scenarios, in this work the closed-form analytical expressions are derived for the coverage probability (CP) over S-CH FSO links and CH-UAV RF links in interference-free, interference-dominated, and interference-and-noise cases. 

3) The asymptotic expressions for the CP are derived and the diversity orders are calculated for both S-CH FSO link and CH-UAV RF link in interference-free case.

The rest of this paper is organized as follows. In Section II, the considered dual-hop cooperative  satellite-UAV communication system is described. In Section III, IV and V, the coverage  performance of S-CH FSO links, CH-UAV RF links, and the end-to-end (e2e) outage performance of S-CH-UAV links is investigated. In Section VI, numerical results are presented and discussed. Finally, the paper is concluded with some remarks in Section VII.
\section{System Model}\label{smodel}
\begin{figure}[!htb]
\centering 
 \setlength{\abovecaptionskip}{0pt}
  \setlength{\belowcaptionskip}{10pt}
\begin{minipage}[b]{0.65\textwidth} 
\centering 
\setlength{\abovecaptionskip}{0pt}
\setlength{\belowcaptionskip}{10pt}
\includegraphics[width=4 in]{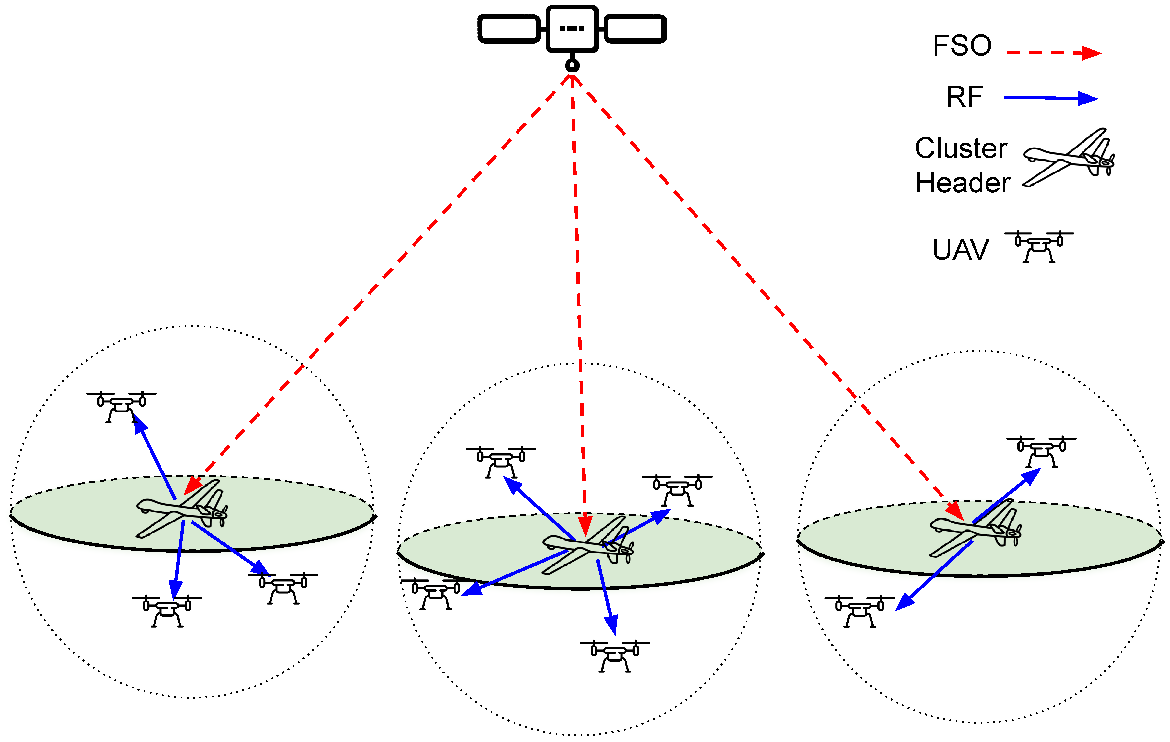}
    \caption{System model}
\label{sm}
\end{minipage}
\begin{minipage}[b]{0.3\textwidth} 
\centering 
\setlength{\abovecaptionskip}{0pt}
\setlength{\belowcaptionskip}{10pt}
\includegraphics[width=1.8 in]{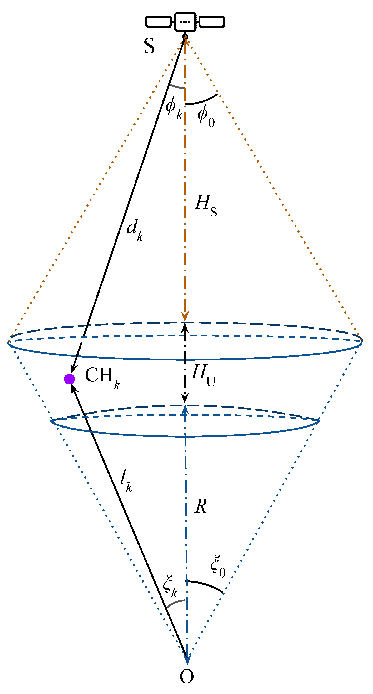}
    \caption{S-CH system model}
\label{sh}
\end{minipage}
\vspace{-10mm}
\end{figure}
In this work, a dual-hop cooperative satellite-UAV communication system, which contains a satellite (S) and a group of cluster headers (CHs)\footnote{In practical, CHs can be common airplanes piloted by human or the UAVs with superior hardware resources, which is capable of providing data receiving, processing, and forwarding functionalities to serve as aerial relays between S and UAVs.} that are respectively with a group of uniformly distributed UAVs, is considered, as shown in Fig.\ref{sm}. Specifically, S first delivers its information to CHs, and then each CH decodes and forwards the recorded information to the UAVs within its serving space. 
\subsection{S-CH FSO link}

It is assumed that FSO communication technology is adopted over S-CH links to exploit its high directivity to minimize the probability that the transmitted information is wiretapped during the long-distance transmission over S-CH links\footnote{Normally, the distance of satellite-aerial transmissions ranges from hundreds of to tens of thousands of kilometers, which depends on the orbit heights of the satellite.},  and RF communication is employed over CH-UAV links to utilize its omnidirectional coverage ability to realize information broadcasting in the local space of each CH.

Also, to reflect and meet the practical aerial scenarios, in the considered system model the locations of CHs in 3D space are assumed to obey an MHCPP type-II, denoted by $\Phi_{\rm CH}$, with an intensity of $\lambda_{\rm CH}$ and a minimum distance $D_{\min}$ between different CHs. To obtain the proposed MHCPP, a three-step thinning process is applied. Firstly, candidate points whose distribution follows an HPPP $\Phi_{\rm P}$ with an intensity $\lambda_{\rm P}$ are generated in the way that these points are uniformly distributed in the considered 3D space $\mathcal{V}$ and the number of candidate points $N_{\rm P}=|\Phi_{\rm P}|$ has the probability mass function as
\begin{align}
    {\bf{Pr}}\{N_{\rm P}=n\}=\frac{(\lambda_{\rm P}V)^n}{n!}\exp{(-\lambda_{\rm P}V)},
\end{align}
where $V$ is the volume of $\mathcal{V}$.

Then, secondly, each candidate point is assigned with an independent mark which obeys uniform distribution ranging from 0 to 1. Thirdly, all the candidate points are one by one checked whether it is associated with the smallest mark compared with all the other points around it within a distance, $D_{\min}$. If true, the point can remain in $\Phi_{\rm CH}$. Otherwise, it will be eliminated. From the above generating process, one can see that each CH exhibits a spherical repulsion space with the radius, $D_{\min}$.

Thus, the relationship between $\lambda_{\rm CH}$ and $\lambda_{\rm P}$ is given as \cite{chiu2013stochastic}
\begin{align}\label{intensity}
        \lambda_{\rm CH}=\frac{1-\exp{\left(-\frac{4}{3}\pi D_{\min}^3\lambda_{\rm P}\right)}}{\frac{4}{3}\pi D_{\min}^3}.
\end{align}

As shown in Fig. \ref{sh}, CHs are distributed inside a 3D space $\mathcal{V}$ that subtracts a spherical cone with radius $R$ from the other one with radius $R+H_{\rm U}$. The two spherical cones share the same center O and apex angle $\xi_0$. The volume of $\mathcal{V}$ is $V=\frac{2\pi}{3}(1-\cos{\xi_0})[(R+H_{\rm U})^3-R^3]$.

In the FSO link, the received electrical signal of CH$_k$ after photoelectric conversion \cite{ansari2016performance} is
\begin{align}
    y_{k}=\eta P_{\rm S} G_{\rm S} G_{\rm R} \left(\frac{\lambda}{4\pi d_k}\right)^2 h_p h_a h_l s_k+n_k,
\end{align}
where $P_{\rm S}$ is the transmit optical power, $\eta$ is the effective photoelectric conversion ratio, $G_{\rm S}$ and $G_{\rm R}$ are the telescope gains of the transmitter and receiver, $\lambda$ is the wavelength of the laser, $d_k$ is the distance from S to CH$_k$, $h_p$ the random attenuation caused by beam spreading and misalignment fading, $h_a$ is the random attenuation caused by atmospheric turbulence, $h_l$ is the atmospheric loss, $s_k$ is the transmitted symbol with the average power of 1, and $n_k$ is the additive white Gaussian noise (AWGN) of CH$_k$ with power $N_{\rm F}$.

In this work, we adopt the fading model mentioned in \cite{zedini2020performance} which considers the atmospheric loss $h_l$, the atmospheric turbulence $h_a$ and the misalignment fading $h_p$. Then, the CDF of the channel power gain $h=h_lh_ah_p$ can be given as 
\begin{align}\label{cdfh}
    F_{h}(x)&=\frac{\omega^2}{\Gamma(\alpha)\Gamma(\beta)}\MeijerG{1}{\omega^2+1}{\omega^2,\alpha,\beta}{0}{\frac{\alpha\beta}{A_0h_l}x},
\end{align}
where $\Gamma(\cdot)$ is the Gamma function, $\alpha$ and $\beta$ are the effective number of small-scale and large-scale eddies of the scattering environment, $G(\cdot|\cdot)$ is the Meijer-$G$ function, $A_0$ is the fraction of power collected by the detector when there is no pointing error, and $\omega$ is the ratio between the equivalent beam radius and the pointing error displacement standard deviation at the receiver.

In this context, the SNR of received signal at CH$_k$ is 
\begin{align}\label{snrkCH}
    \gamma_k=\frac{\eta^2 P^2_{\rm S} G^2_{\rm S} G^2_{\rm R} \lambda^4  h^2}{(4\pi)^4 d_k^4 N_{\rm F}}.
\end{align}

\begin{lemma}
Given $d_k$, the CDF of $\gamma_k$ is 
\begin{align}\label{cdfsnr}
    F_{\gamma_k|d_k^2}(x)=\frac{\omega^2}{\Gamma(\alpha)\Gamma(\beta)}\MeijerG{1}{\omega^2+1}{\omega^2,\alpha,\beta}{0}{\Xi d_k^2 \sqrt{x}},
\end{align}
where $\Xi=\frac{ \alpha \beta(4\pi)^2\sqrt{N_{\rm F}}}{A_0 h_l\eta \lambda^2 P_{\rm S} G_{\rm S} G_{\rm R} }$.
\end{lemma}
\begin{proof}
From \eqref{cdfh} and \eqref{snrkCH}, we can get
\begin{align}
    F_{\gamma_k|d_k^2}(x)&={\bf{Pr}}\{\gamma_k \leq x\}={\bf{Pr}}\left\{\frac{\eta^2 P^2_{\rm S} G^2_{\rm S} G^2_{\rm R} \lambda^4  h^2}{(4\pi)^4 d_k^4 N_{\rm F}}\leq x\right\}\notag\\
    &={\bf{Pr}}\left\{h\leq \frac{(4\pi)^2 d_k^2}{\eta \lambda^2}\frac{\sqrt{N_{\rm F}x}}{P_{\rm S}G_{\rm S}G_{\rm R}}\right\}=\frac{\omega^2}{\Gamma(\alpha)\Gamma(\beta)}\MeijerG{1}{\omega^2+1}{\omega^2,\alpha,\beta}{0}{\Xi d_k^2 \sqrt{x}},
\end{align}
where $\Xi=\frac{ \alpha \beta(4\pi)^2\sqrt{N_{\rm F}}}{A_0 h_l\eta \lambda^2 P_{\rm S} G_{\rm S} G_{\rm R} }$.
\end{proof}
 
\textcolor{black}{
\begin{proposition}
When $D_{\min}\ll ({\rm V})^{\frac{1}{3}}$, CHs are approximately independently and uniformly distributed in $\mathcal{V}$.
\end{proposition}
\begin{proof}
Please refer to Appendix I.
\end{proof}
}

 \begin{lemma}\label{lemmapdfdk2}
 The PDF of $d_k^2$ is
 \begin{align}\label{pdfdk2}
     f_{d_k^2}(x)=\frac{\pi}{2VL}\left[\tau_2^2(x)-\tau_1^2(x)\right],
 \end{align}
where $d_{\min}^2\leq x\leq d_{\max}^2$, $d_{\min}=H_{\rm S}$, $d_{\max}=\sqrt{(R+H_{\rm U})^2+L^2-2(R+H_{\rm U})L\cos{\xi_0}}$, $L=H_{\rm S}+H_{\rm U}+R$, $\tau_2(x)=\min{\left\{R+H_{\rm U},L\cos{\xi_0}-\sqrt{x-L^2\sin^2{\xi_0}}\right\}}$, and $\tau_1(x)=\max{\left\{R,L-\sqrt{x} \right\}}$.
 \end{lemma}
 \begin{proof}
 Please refer to Appendix II.
 \end{proof}
 

\begin{theorem}
The CDF of $\gamma_k$ is
\begin{align}
    F_{\gamma_k}(x)&=\frac{\pi^2 b_1\omega^2}{2M_fVL\Gamma(\alpha)\Gamma(\beta)}\sum_{i=1}^{M_f}\left[\tau_2^2(b_i)-\tau_1^2(b_i)\right]\sqrt{1-t^2_i}\MeijerG{1}{\omega^2+1}{\omega^2,\alpha,\beta}{0}{\Xi b_i\sqrt{x}},
\end{align}
where $t_i=\cos{\frac{2i-1}{2M_f}\pi}$, $b_i=b_1t_i+b_2$, $b_1=\frac{d^2_{\max}-d^2_{\min}}{2}$ and $b_2=\frac{d^2_{\max}+d^2_{\min}}{2}$.
\end{theorem} 
 
\begin{proof}
Using \eqref{cdfsnr}, the CDF of $\gamma_k$ is
\begin{small}
\begin{align}\label{cdfgammak}
    F_{\gamma_k}(x)&=\int\limits_{d^2_{\min}}^{d^2_{\max}}F_{\gamma_k|d_k^2}(x|y)f_{d_k^2}(y) dy=\int\limits_{d^2_{\min}}^{d^2_{\max}}\frac{\pi \omega^2 }{2VL\Gamma(\alpha)\Gamma(\beta)}[\tau_2^2(y)-\tau_1^2(y)]\MeijerG{1}{\omega^2+1}{\omega^2,\alpha,\beta}{0}{\Xi x\sqrt{y}}dy.
\end{align}
\end{small}

By setting $b_1=\frac{d^2_{\max}-d^2_{\min}}{2}$, $b_2=\frac{d^2_{\max}+d^2_{\min}}{2}$ and $t=\frac{y-b_2}{b_1}$ and employing Chebyshev-Gauss quadrature in the first
case, \eqref{cdfgammak} can be written as
\begin{align}
 F_{\gamma_k}(x) &=\frac{\pi b_1\omega^2}{2VL\Gamma(\alpha)\Gamma(\beta)}\int\limits_{-1}^{1}\left[\tau_2^2(b_1t+b_2)-\tau_1^2(b_1t+b_2)\right] \MeijerG{1}{\omega^2+1}{\omega^2,\alpha,\beta}{0}{\Xi (b_1t+b_2)\sqrt{x}}dt\notag\\
    &=\frac{\pi^2 b_1\omega^2}{2M_fVL\Gamma(\alpha)\Gamma(\beta)}\sum_{i=1}^{M_f}\left[\tau_2^2(b_i)-\tau_1^2(b_i)\right]\sqrt{1-t^2_i} \MeijerG{1}{\omega^2+1}{\omega^2,\alpha,\beta}{0}{\Xi b_i\sqrt{x}},
\end{align}
where $t_i=\cos{\frac{2i-1}{2M_f}\pi}$ and $b_i=b_1t_i+b_2$.
\end{proof}

\subsection{CH-UAV RF link}

Moreover, in this work, we also assume that UAVs around the $k$th CH, CH$_k$, are uniformly distributed in its serving sphere, which is centered at CH$_k$ with radius $D_k$, and their positions in 3D space obey an HPPP $\Phi_{k{\rm U}}$ with an intensity $\lambda_{k{\rm U}}$. CH$_k$ forwards the recoded information with the transmit power $P_k$ to the UAVs within its serving space.

In the second hop shown in Fig. \ref{sm}, CH$_k$ will transmit the recoded signals to the UAVs within its serving space, namely, the ones within the sphere with radius $D_k$ centered at CH$_k$. 

To guarantee that there is no intersection between serving spaces of any two CHs (CH$_k$ and CH$_j$), $D_j+D_k<D_{\min}$ ($j \neq k$) should be satisfied. The number of UAVs around CH$_k$, $N_{k{\rm U}}$, follows an HPPP with the density $\lambda_{k{\rm U}}$. The PMF of $N_{k{\rm U}}$ is ${\bf{Pr}}\{N_{k{\rm U}}=n\}=(\mu^n/n!)\exp{(-\mu)}$, where $\mu=\frac{4\pi D_k^3 }{3}\lambda_{k{\rm U}}$ is the mean measure. To make the analysis tractable, we assume that all CHs have the same serving radius and the UAVs around them have the same density, namely, $D_k=D$ and $\lambda_{k{\rm U}}=\lambda_{\rm U}$ ($k=1,..,N_{\rm CH}$, $N_{\rm CH}=|\Phi_{\rm CH}|$). Meanwhile, we assume that the communication channels from CH to UAVs suffer Nakagami-$m$ fading\footnote{As suggested in Chapter 3.2.3 of \cite{goldsmith2005wireless}, Nakagami-$m$ can become approximate Rician fading with parameter $K$ which can be deduced by the fading parameter $m=(K+1)^2/(2K+1)$. In other words, Nakagami-$m$ can represent the channel fading in the case of LoS transmission, which is the typical propagation scenarios for the transmissions between CHs and UAVs. Also, when $m$ approaches infinite, Nakagami-$m$ can be used to describe the case that there is no fading.}.

Then, the PDF and CDF of channel power gain $g$ in this case are shown as
\begin{align}\label{nakagamipdf}
    f_g(x)=\left(\frac{m}{\Omega}\right)^{m}\frac{x^{m-1}}{(m-1)!}\exp{\left(-\frac{m}{\Omega}x\right)}
\end{align}
and
\begin{align}\label{nakagamicdf}
    F_g(x)=1-\sum\limits_{i=0}^{m-1}\left(\frac{m}{\Omega}\right)^{m_i}\frac{x^{m_i}}{m_i!}\exp{\left(-\frac{m}{\Omega}x\right)},
\end{align}
respectively, where $\Omega$ is the average received power, $m$ is the fading parameter, and $m_i=m-i-1$. Notably, to simplify the analysis, we only consider the case that $m$ is integer in the following of this paper.  

The free-space path-loss from CH$_k$ to the $j$th $(0 < j \leq N)$ UAV marked as U$_j$ can be given by $\rho d_{kj}^{\alpha_r}$, where $\rho$ denotes the path-loss at the distance $d = 1$ m and its value depends on the carrier frequency, $\alpha_r$ is the path-loss factor, and $d_{kj}$ is the link distance between CH$_k$ and U$_j$.

\begin{lemma}
The PDF of $d_{kj}$ are respectively given as
\begin{align}\label{pdfdkj}
    f_{d_{kj}}(x)=\left\{
\begin{array}{ll}
\frac{3}{D^3}x^2, & 0\leq x \leq D;\\
0, & {\rm{otherwise}}
\end{array}
\right..
\end{align}
\end{lemma}
\begin{proof}
UAVs served by CH$_k$ can be modeled as a set of independently and identically uniformly distributed points in a sphere centered at CH$_k$, denoted as $W_k$. According to \cite{pan20173d}, $d_{kj}$ can be calculated from $W_k$, the PDF of which can be presented as 
\begin{align}
    f_{W_k}=\frac{\lambda_{\rm U}}{\mu}=\frac{3}{4\pi D^3}.
\end{align}

Therefore, the CDF of $d_{kj}$ can be calculated as
\begin{align}
    F_{d_{kj}}(x)&=\int\limits_{0}^{x}\int\limits_{0}^{\pi}\int\limits_{0}^{2\pi}\frac{3}{4\pi D^3} \sin{\xi}r^2d\theta d\xi dr=\left\{
\begin{array}{ll}
0, & x< 0;\\
\frac{x^3}{D^3}, & 0\leq x \leq D;\\
1, & D<x
\end{array}
\right..
\end{align}

Then, the PDF of $d_{kj}$ can be obtained though $f_{d_{kj}}(x)=\frac{dF_{d_{kj}}(x)}{dx}$.
\end{proof}

\section{Performance Analysis over S-CH FSO Links}

\subsection{Coverage Performance}
CP is defined as the ergodic probability that the received SNR of a randomly selected receiver exceeds a specific threshold. Adopting \eqref{cdfsnr} and given the SNR threshold $\gamma_{\rm th}$, the CP of CH$_k$ is given as
\begin{align}\label{pcop}
   & P_{{\rm cov},{\rm{SCH}}_k}(\gamma_{\rm th}) ={\bf{Pr}}\{\gamma_k>\gamma_{\rm th}\}=1-F_{\gamma_k}(\gamma_{\rm th})\notag\\
    &~~~=1-\frac{\pi^2 b_1\omega^2}{2M_fVL\Gamma(\alpha)\Gamma(\beta)}\sum_{i=1}^{M_f}\left[\tau_2^2(b_i)-\tau_1^2(b_i)\right]\sqrt{1-t^2_i} \MeijerG{1}{\omega^2+1}{\omega^2,\alpha,\beta}{0}{\Xi b_i\sqrt{\gamma_{\rm th}}}.
\end{align}

\subsection{Asymptotic Coverage Performance}

\begin{theorem}
The CP at high transmit SNR ($P_{\rm S}^2/N_{\rm F}\rightarrow\infty$) can be expressed as 
\begin{small}
\begin{align}\label{cpopasym}
    P^{\infty}_{{{\rm cov},{\rm{SCH}}}_k}(\gamma_{\rm th}) \approx&1-\frac{\pi^2 b_1\omega^2}{2M_fVL\Gamma(\alpha)\Gamma(\beta)}\sum_{i=1}^{M_f}\sqrt{1-t^2_i} \left[\tau_2^2(b_i)-\tau_1^2(b_i)\right]\Bigg[\left(\Xi b_i\sqrt{\gamma_{\rm th}}\right)^{\omega^2}\frac{\Gamma(\alpha-\omega^2)\Gamma(\beta-\omega^2)}{\omega^2}\notag\\
    & +\left(\Xi b_i\sqrt{\gamma_{\rm th}}\right)^{\alpha}\frac{\Gamma(\omega^2-\alpha)\Gamma(\beta-\alpha)}{\alpha\Gamma(\omega^2+1-\alpha)}+\left(\Xi b_i\sqrt{\gamma_{\rm th}}\right)^{\beta}\frac{\Gamma(\omega^2-\beta)\Gamma(\alpha-\beta)}{\beta\Gamma(\omega^2+1-\beta)}\Bigg]
\end{align}
\end{small}

\end{theorem}

\begin{proof}
By first inverting the argument of the Meijer-$G$ function using \cite[Eq. (9.31.2)]{gradshteyn2014table} and \cite[Eq.(17)]{ansari2015performance}, 
\eqref{pcop} can be approximated at high SNR ($\Xi\rightarrow0$) as \eqref{cpopasym}.
\end{proof}

\subsection{Diversity Order}
In this work, the diversity order of the considered system is defined as
\begin{align}\label{diversity}
    \Lambda=-\lim_{\Bar{\gamma}\rightarrow\infty}\frac{\log (1-P_{{\rm cov}})}{\log \Bar{\gamma}},
\end{align}
where $\Bar{\gamma}$ is the average transmit SNR and $P_{{\rm cov}}$ is the CP.

\begin{corollary}
The diversity order of S-CH FSO link is $\min\{\omega^2,\alpha,\beta\}$.
\end{corollary}

\begin{proof}
     One can see that there are three additive terms in the square brackets $[\cdot]$ in \eqref{cpopasym}. When $\Xi\rightarrow 0$, it is obvious that the term with the minimum power index is dominant. According to \eqref{diversity}, the diversity order can be reached as $\min\{\omega^2,\alpha,\beta\}$.
\end{proof}

\section{Performance Analysis over CH-UAV RF Links}\label{rflink}

\subsection{Interference-Free Case}
We will first consider the case that there is no interference from other CHs, which represents the cases that the interfering CHs are too far or the transmit power at the interfering CH is too low to incur effective interference to the target UAV. 

Supposing that CH$_k$ has the transmit power of $P_{\rm R}$, the SNR at U$_j$ in the serving space of CH$_k$ can be expressed as
\begin{align}
    \gamma_{kj}=\frac{P_{\rm R}g_{kj}}{\rho d_{kj}^{\alpha_r}N_{\rm R}},
\end{align}
where $g_{kj}$ is the channel power gain of CH$_k$-U$_j$ link and $N_{\rm R}$ is the average power of the AWGN at U$_j$.

\begin{theorem}
The CDF of $\gamma_{kj}$ can be calculated as
\begin{align}\label{cdfrkj}
    F_{\gamma_{kj}}(x)=&1-\frac{3}{\alpha_rD^3}\left(\frac{ \Omega P_{\rm R}}{m\rho N_{\rm R} }\frac{1}{x}\right)^{\frac{3}{\alpha_r}}\sum\limits_{i=0}^{m-1}\frac{1}{m_i!}\gamma\left(m_i+\frac{3}{\alpha_r},\frac{ m\rho N_{\rm R} D^{\alpha_r} }{\Omega P_{\rm R}}x\right).
\end{align}
\end{theorem}

\begin{proof}
Using \eqref{nakagamicdf} and \eqref{pdfdkj}, $F_{\gamma_{kj}}(z)$ can be expressed as 
\begin{align}\label{cdfrkjnoin}
    F_{\gamma_{kj}}(x)&={\bf{Pr}}\{\gamma_{kj}<x\}={\bf{Pr}}\left\{\frac{P_{\rm R}g_{kj}}{\rho d_{kj}^{\alpha_r}N_{\rm R}}<x\right\}={\bf{Pr}}\left\{g_{kj}<\frac{\rho d_{kj}^{\alpha_r}N_{\rm R}}{P_{\rm R}}x\right\}\notag\\
    &=\int\limits_{0}^{D}F_g\left(\frac{\rho y^{\alpha_r}N_{\rm R}}{P_{\rm R}}x\right)f_{d_{kj}}(y)dy\notag\\
    &=1-\frac{3}{D^3}\sum\limits_{i=0}^{m-1}\left(\frac{m\rho N_{\rm R}}{\Omega P_{\rm R}}x\right)^{m_i}\frac{1}{m_i!} \underbrace{\int\limits_{0}^{D}y^{\alpha_r m_i+2}\exp{\left(-\frac{m\rho N_{\rm R}}{\Omega P_{\rm R}}xy^{\alpha_r}\right)}dy}_{I_0}.
\end{align}

Resorting to \cite[Eq. (3.381.8)]{gradshteyn2014table}, $I_0$ can be given as 
\begin{align}\label{I0}
    I_0=\frac{1}{\alpha_r}\left(\frac{ \Omega P_{\rm R}}{m\rho N_{\rm R} x}\right)^{m_i+\frac{3}{\alpha_r}}\gamma\left(m_i+\frac{3}{\alpha_r},\frac{ m\rho N_{\rm R} D^{\alpha_r} }{\Omega P_{\rm R}}x\right),
\end{align}
where $\gamma(\cdot,\cdot)$ is the lower incomplete gamma function.

Substituting \eqref{I0} in \eqref{cdfrkjnoin}, the CDF of $\gamma_{kj}$ can be obtained as \eqref{cdfrkj}.
\end{proof}

Then, the CP in this case can be achieved as
\begin{align}\label{cpkj5}
    P_{{\rm cov},kj}(\gamma_{\rm th})&=1-F_{\gamma_{kj}}(\gamma_{\rm th})=\frac{3}{\alpha_rD^3}\left(\frac{ \Omega P_{\rm R}}{m\rho N_{\rm R} \gamma_{\rm th}}\right)^{\frac{3}{\alpha_r}}\sum\limits_{i=0}^{m-1}\frac{1}{m_i!}\gamma\left(m_i+\frac{3}{\alpha_r},\frac{ m\rho N_{\rm R} D^{\alpha_r} \gamma_{\rm th}}{\Omega P_{\rm R}}\right).
\end{align}
\subsubsection{Asymptotic Coverage Performance}
In high SNR regime, the CDF of $g$ in Nakagami-$m$ fading is given as \cite{lei2013outage} 
\begin{align}\label{asymnakagamicdf}
    F^{\infty}_g(x)=\frac{m^{m-1}}{(m-1)!}\left(\frac{x}{\Omega}\right)^{m}.
\end{align}

Substituting \eqref{asymnakagamicdf} in \eqref{cdfrkjnoin}, we can get the asymptotic coverage probability $P^{\infty}_{{\rm cov},kj}(\gamma_{\rm th})$ as
\begin{align}\label{asympcpnoise}
     P^{\infty}_{{\rm cov},kj}(\gamma_{\rm th})&=1-\int\limits_{0}^{D}F^{\infty}_g\left(\frac{\rho N_{\rm R}\gamma_{\rm th}}{P_{\rm R}}y^{\alpha_r}\right)f_{d_{kj}}(y)dy\notag\\
     &=1-\frac{3}{D^3}\frac{m^{m-1}}{(m-1)!}\left(\frac{\rho N_{\rm R}\gamma_{\rm th}}{\Omega P_{\rm R}}\right)^{m}\int\limits_{0}^{D}y^{\alpha_rm+2}dy\notag\\
     &=1-\frac{3m^{m-1}}{(m-1)!}\left(\frac{\rho N_{\rm R}\gamma_{\rm th}}{\Omega P_{\rm R}}\right)^{m}\frac{D^{\alpha_rm}}{\alpha_rm+3}.
\end{align}
\subsubsection{Diversity Order}
\begin{corollary}
The diversity order of CH-UAV links in interference-free case is $m$.
\end{corollary}
\begin{proof}
     From \eqref{diversity} and \eqref{asympcpnoise}, it is easy to obtain the diversity order as $m$.
\end{proof}
\subsection{Interference-Dominated Case}

As the UAV's receiver has limited sensitivity, we consider that U$_j$ is only disturbed by these CHs within the distance of $D_{\max}$ (usually, $D_{\max}\gg D$ and $D_{\max}\gg D_{\min}$).

To simplify the analysis, we assume that all CHs have the same transmit power $P_{\rm R}$ and the channel power gains between interfering CHs and U$_j$ are independent and identically distributed random variables with parameters $m$ and $\Omega$. As the interfering power is much greater than the noise power \cite{hunter2008trans}, signal-to-interference ratio (SIR) is considered in this case.

The received SIR at U$_j$ around CH$_k$ can be presented as  
\begin{align}
    \gamma_{kj}=\frac{\frac{P_{\rm R}g_{kj}}{\rho d_{kj}^{\alpha_r}}}{\sum\limits_{i=1}^{N_I}\frac{P_{\rm R}g_{jI_i}}{\rho d_{jI_i}^{\alpha_r}}}=\frac{\frac{g_{kj}}{ d_{kj}^{\alpha_r}}}{I},
\end{align}
where $I=\sum\limits_{i=1}^{N_I}\frac{g_{jI_i}}{ d_{jI_i}^{\alpha_r}}$, $g_{jI_i}$ is the channel power gain between the $I_i$th interfering CH and U$_j$, and $d_{jI_i}$ is the distance from the $I_i$th CH to U$_j$.

\subsubsection{Coverage Performance}
The CP of the RF link in this case can be written as 
\begin{align}\label{cprf}
    P_{{\rm cov},kj}(\gamma_{\rm th})&=\mathbb{E}_{I,d_{kj}}\left[{\bf{Pr}}\{\gamma_{kj}\geq \gamma_{\rm th}|I, d_{kj}\}\right]=\mathbb{E}_{I,d_{kj}}\left[{\bf{Pr}}\{g_{kj}\geq \gamma_{\rm th}Id_{kj}^{\alpha_r}|I, d_{kj}\}\right].
\end{align}

\begin{lemma}
$P_{{\rm cov},kj}(\gamma_{\rm th})$ can be expressed as
\begin{align}\label{cprf1}
     P_{{\rm cov},kj}(\gamma_{\rm th})=& \sum\limits_{i=0}^{m-1}\left(-\frac{m\gamma_{\rm th}}{\Omega}\right)^{m_i}\frac{1}{m_i!} \mathbb{E}_{d_{kj}}\left\{d_{kj}^{\alpha_rm_i}\frac{d^{m_i}\mathbb{E}_{I}[e^{-sI}]}{ds^{m_i}}\right\},
\end{align}
where $s=m \gamma_{\rm th}d_{kj}^{\alpha_r}/\Omega$.
\end{lemma}

\begin{proof}
According to \eqref{nakagamicdf}, \eqref{cprf} can be calculated as
\begin{align}\label{pcovkj}
     P_{{\rm cov},kj}(\gamma_{\rm th})&=\mathbb{E}_{I,d_{kj}}\Bigg[\sum\limits_{i=0}^{m-1}\left(\frac{m}{\Omega}\right)^{m_i}\frac{(\gamma_{\rm th}Id_{kj}^{\alpha_r})^{m_i}}{m_i!}\exp{\left(-\frac{m\gamma_{\rm th}Id_{kj}^{\alpha_r}}{\Omega}\right)}\Bigg]\notag\\
     &=\sum\limits_{i=0}^{m-1}\left(\frac{m}{\Omega}\right)^{m_i}\frac{\gamma_{\rm th}^{m_i}}{m_i!}\mathbb{E}_{I,d_{kj}}\left[(Id_{kj}^{\alpha_r})^{m_i}\exp{\left(-\frac{m\gamma_{\rm th}Id_{kj}^{\alpha_r}}{\Omega}\right)}\right]\notag\\
     &=\sum\limits_{i=0}^{m-1}\left(\frac{m}{\Omega}\right)^{m_i}\frac{\gamma_{\rm th}^{m_i}}{m_i!}\mathbb{E}_{d_{kj}}\left\{d_{kj}^{\alpha_rm_i}\underbrace{\mathbb{E}_{I}\left[I^{m_i}\exp{\left(-\frac{m\gamma_{\rm th}d_{kj}^{\alpha_r}}{\Omega}I\right)}\right]}_{I_2}\right\}.
\end{align}

By setting  $s=\frac{m\gamma_{\rm th}d_{kj}^{\alpha_r}}{\Omega}$, $I_2$ can be obtained as 
\begin{align}\label{I2}
    I_2=\mathbb{E}_{I}\left[I^{m_i}e^{-sI}\right].
\end{align}

From the definition of Laplace transform (LT), the LT of $I$ can be given as $\mathcal{L}_{I}(s)=\int_{0}^{\infty}e^{-st}f_I(t)dt=\mathbb{E}_{I}[e^{-sI}]$, where $f_I(t)$ is the PDF of $I$.

By using the differential property of LT, $I_2$ can be achieved as 
\begin{align}\label{I21}
     I_2&=\int\limits_{0}^{\infty}t^{m_i}e^{-st}f_I(t)dt=(-1)^{m_i}\frac{d^{m_i}\mathcal{L}_I(s)}{ds^{m_i}}=(-1)^{m_i}\frac{d^{m_i}\mathbb{E}_{I}[e^{-sI}]}{ds^{m_i}}.
\end{align}

Combining \eqref{pcovkj} and \eqref{I21}, \eqref{cprf1} can be obtained.

\end{proof}

\begin{lemma}\label{lemmaesi}
$\mathbb{E}_{I}\left[e^{-sI}\right]$ can be expressed as
\begin{align}\label{esi}
    \mathbb{E}_{I}\left[e^{-sI}\right]=\exp{\left[\lambda_{\rm CH} V_1(I_3-1)\right]},
\end{align}
where $I_3$ is present as
\begin{small}
\begin{align}\label{I3}
    &I_3=\left(\frac{m}{\Omega s}\right)^{m}\frac{\pi }{2d_{kj}V_1 }\Bigg[\mathcal{F}\left(\frac{4}{\alpha_r},d_{jI_1}^{\rm{g_1}},d_{jI_1}^{\min}\right)+2d_{kj}\mathcal{F}\left(\frac{3}{\alpha_r},d_{jI_1}^{\rm{g_1}},d_{jI_1}^{\min}\right)+(d_{kj}^2-D_{\min}^2)\mathcal{F}\left(\frac{2}{\alpha_r},d_{jI_1}^{\rm{g_1}},d_{jI_1}^{\min}\right)+4d_{kj}\notag\\
    &\cdot \mathcal{F}\left(\frac{3}{\alpha_r},d_{jI_1}^{\rm{g_2}},d_{jI_1}^{\rm{g_1}}\right)-\mathcal{F}\left(\frac{4}{\alpha_r},d_{jI_1}^{\max},d_{jI_1}^{\rm{g_2}}\right)+2d_{kj}\mathcal{F}\left(\frac{3}{\alpha_r},d_{jI_1}^{\max},d_{jI_1}^{\rm{g_2}}\right)+(D_{\max}^2-d_{kj}^2)\mathcal{F}\left(\frac{2}{\alpha_r},d_{jI_1}^{\max},d_{jI_1}^{\rm{g_2}}\right)\Bigg],
\end{align}
\end{small}

 $d_{jI_1}^{\min}=(D_{\min}-d_{kj})^2$, $d_{jI_1}^{\rm{g_1}}=(D_{\min}+d_{kj})^2$, $d_{jI_1}^{\rm{g_2}}=(D_{\max}-d_{kj})^2$,$d_{jI_1}^{\max}=(D_{\max}+d_{kj})^2$, $V_1=\frac{4\pi}{3}(D_{\max}^3-D_{\min}^3)$, 
   $\mathcal{F}(a,b,c)=\big[\Hypergeometric{2}{1}{m,m+a}{m+a+1}{-\frac{m}{\Omega s}b^{\frac{\alpha_r}{2}}}b^{\frac{(m+a)\alpha_r}{2}}-\Hypergeometric{2}{1}{m,m+a}{m+a+1}{-\frac{m}{\Omega s}c^{\frac{\alpha_r}{2}}}c^{\frac{(m+a)\alpha_r}{2}}\big]\frac{2}{(m+a)\alpha_r}$,  and $\Hypergeometric{2}{1}{\cdot,\cdot}{\cdot}{\cdot}$ denotes Gauss hypergeometric function.
\end{lemma}

\begin{proof}Please refer to Appendix III
\end{proof}

\begin{lemma}\label{lemmai2}
$\frac{d^{m_i}\mathbb{E}_{I}[e^{-sI}]}{ds^{m_i}}$ can be represented as 
\begin{align}\label{iesi}
   \frac{d^{m_i}\mathbb{E}_{I}[e^{-sI}]}{ds^{m_i}}=&\exp{\left(\mathcal{A}(s,d_{kj})\right)}\bigg[1 +\mathbbm{1}\{m_i>0\}\sum\limits_{n=1}^{m_i}\mathcal{B}(m_i,n,\mathcal{A}(s,d_{kj}))\bigg],
\end{align}
where $\mathcal{A}(s,d_{kj})=\lambda_{\rm CH} V_1(I_3-1)$, $\mathcal{B}(m_i,n,\mathcal{A}(s,d_{kj}))=B_{m_i,n}\left(\mathcal{A}^{(1)}(s,d_{kj}),...,\mathcal{A}^{(m_i-n+1)}(s,d_{kj})\right)$, $B_{m_i,n}(\cdot)$ is the Bell polynomials, and $\mathcal{A}^{(n)}(s,d_{kj})$ is the $n$th derivative of $\mathcal{A}(s,d_{kj})$. 
\end{lemma}

\begin{proof}
When $m_i=0$, 
\begin{align}
    \frac{d^{0}}{ds^{0}}\exp{\left(\mathcal{A}(s,d_{kj})\right)}=\exp{\left(\mathcal{A}(s,d_{kj})\right)}.
\end{align}

When $m_i>0$, according to Faà di Bruno's formula, we can obtain 
\begin{align}\label{I5}
    \frac{d^{m_i}}{ds^{m_i}}\exp{\left(\mathcal{A}(s,d_{kj})\right)}=&\exp{\left(\mathcal{A}(s,d_{kj})\right)}\times\sum\limits_{n=1}^{m_i}\mathcal{B}(m_i,n,\mathcal{A}(s,d_{kj})),
\end{align}
where $\mathcal{B}(m_i,n,\mathcal{A}(s,d_{kj}))=B_{m_i,n}\left(\mathcal{A}^{(1)}(s,d_{kj}),...,\mathcal{A}^{(m_i-n+1)}(s,d_{kj})\right)$, $B_{m_i,n}(\cdot)$ is the Bell polynomials, and $\mathcal{A}^{(n)}(s,d_{kj})$ is the $n$th derivative of $\mathcal{A}(s,d_{kj})$.

\end{proof}

\begin{lemma}\label{lemmaans}
$\mathcal{A}^{(n)}(s,d_{kj})$ ($n>0$) can be represented as 
\begin{align}\label{nAq}
     &\mathcal{A}^{(n)}(s,d_{kj})= \frac{\pi\lambda_{\rm CH} }{2d_{kj}}\left(\frac{m}{\Omega}\right)^{m}\sum\limits_{l=0}^{n}\binom{n}{l} (-1)^{n-l}(m)_{n-l}s^{-m-n+l} \Bigg[\mathcal{F}^{(l)}\left(\frac{4}{\alpha_r},d_{jI_1}^{\rm{g_1}},d_{jI_1}^{\min}\right)\notag\\
     &+2d_{kj}\mathcal{F}^{(l)}\left(\frac{3}{\alpha_r},d_{jI_1}^{\rm{g_1}},d_{jI_1}^{\min}\right)+(d_{kj}^2-D_{\min}^2)\mathcal{F}^{(l)}\left(\frac{2}{\alpha_r},d_{jI_1}^{\rm{g_1}},d_{jI_1}^{\min}\right)+4d_{kj}\mathcal{F}^{(l)}\left(\frac{3}{\alpha_r},d_{jI_1}^{\rm{g_2}},d_{jI_1}^{\rm{g_1}}\right)\notag\\
    &-\mathcal{F}^{(l)}\left(\frac{4}{\alpha_r},d_{jI_1}^{\max},d_{jI_1}^{\rm{g_2}}\right)+2d_{kj}\mathcal{F}^{(l)}\left(\frac{3}{\alpha_r},d_{jI_1}^{\max},d_{jI_1}^{\rm{g_2}}\right)+(D_{\max}^2-d_{kj}^2)\mathcal{F}^{(l)}\left(\frac{2}{\alpha_r},d_{jI_1}^{\max},d_{jI_1}^{\rm{g_2}}\right)\Bigg],
\end{align}

where 
$$\mathcal{F}^{(l)}(a,b,c)=\left\{
\begin{array}{ll}
\mathcal{F}(a,b,c)&\text{if }l=0\\
\begin{array}{@{}ll@{}}
     \frac{2}{(m+a)\alpha_r}\Big[b^{\frac{(m+a)\alpha_r}{2}}\Delta^{(l)}(s,a,b)  \\
      -c^{\frac{(m+a)\alpha_r}{2}}\Delta^{(l)}(s,a,c)\Big]
\end{array}
 &\text{if } l>0\\
\end{array} \right.,
$$
\begin{align}
    \Delta^{(l)}(s,a,b)=&\sum\limits_{q=1}^{l} \frac{(m+a)(m)_{q}}{m+a+q}\Hypergeometric{2}{1}{m+q,m+a+q}{m+a+q+1}{-\frac{mb^{\frac{\alpha_r}{2}}}{\Omega s}}\notag\\
    &~~~~\times B_{l,q}\left(f_3^{(1)}(s,b),...,f_3^{(l-q+1)}(s,b)\right)\notag,
\end{align}
$(m)_{q}=\prod\limits_{k=0}^{q-1}(m-k)$ is the rising Pochhammer symbol, 
and $f_3^{(q)}(s,b)=(-1)^{q+1}q!\frac{mb^{\frac{\alpha_r}{2}}}{\Omega }s^{-1-q}$.

\end{lemma}

\begin{proof}
Please refer to Appendix IV.
\end{proof}

\begin{theorem}
The CP of the RF link in the interfering-dominated case can be expressed as 
\begin{align}\label{cprf3}
     P_{{\rm cov},kj}(\gamma_{\rm th})&\approx \frac{3\pi }{2D^2M_r}\sum\limits_{p=1}^{M_r}\sqrt{1-t_p^2}\sum\limits_{i=0}^{m-1}\left(-\frac{m\gamma_{\rm th}}{\Omega}\right)^{m_i}\frac{b_p^{\alpha_rm_i+2}}{m_i!} \exp{\left(\tilde{\mathcal{A}}(\gamma_{\rm th},b_p)\right)}\notag\\ 
     &\times\Bigg[1+    \mathbbm{1}\{m_i>0\} \sum\limits_{n=1}^{m_i} \mathcal{B}\left(m_i,n,\tilde{\mathcal{A}}(\gamma_{\rm th},b_p)\right)\Bigg],
\end{align}
where $\tilde{\mathcal{A}}(\gamma_{\rm th},b_p)=\mathcal{A}\left(\frac{m\gamma_{\rm th}b_p^{\alpha_r}}{\Omega},b_p\right)$, $t_p=\cos{\left(\frac{2p-1}{2M_r}\pi\right)}$, $b_p=\frac{D}{2}(t_p+1)$, $\mathcal{A}(\cdot,\cdot)$ is defined in Lemma \ref{lemmai2} and \eqref{I3}, and $\mathcal{B}(\cdot,\cdot,\mathcal{A}(\cdot,\cdot))$ is defined in Lemma \ref{lemmai2} and \eqref{nAq}.
\end{theorem}
\begin{proof}
Substituting \eqref{pdfdkj} and \eqref{iesi} in \eqref{cprf1}, and then substituting $s=m\gamma_{\rm th}d_{kj}^{\alpha_r}/\Omega$, $P_{{\rm cov},kj}(\gamma_{\rm th})$ can expressed as
\begin{small}
\begin{align}\label{cpkjdkj}
   & P_{{\rm cov},kj}(\gamma_{\rm th})=\sum\limits_{i=0}^{m-1}\left(-\frac{m\gamma_{\rm th}}{\Omega}\right)^{m_i}\frac{1}{m_i!} \mathbb{E}_{d_{kj}}\bigg\{d_{kj}^{\alpha_rm_i}\exp{\left(\mathcal{A}(s,d_{kj})\right)}\bigg[1 +\mathbbm{1}\{m_i>0\}\sum\limits_{n=1}^{m_i}\mathcal{B}(m_i,n,\mathcal{A}(s,d_{kj}))\bigg]\bigg\}\notag\\
    &~~~=\sum\limits_{i=0}^{m-1}\left(-\frac{m\gamma_{\rm th}}{\Omega}\right)^{m_i}\frac{1}{m_i!}\frac{3}{D^3}\int\limits_{0}^{D}x^{\alpha_rm_i+2}\exp{\left(\tilde{\mathcal{A}}(\gamma_{\rm th},x)\right)}\bigg[1 +\mathbbm{1}\{m_i>0\}\sum\limits_{n=1}^{m_i}\mathcal{B}\left(m_i,n,\tilde{\mathcal{A}}(\gamma_{\rm th},x)\right)\bigg]dx,
\end{align}
\end{small}
where $\tilde{\mathcal{A}}(\gamma_{\rm th},x)=\mathcal{A}\left(\frac{m\gamma_{\rm th}x^{\alpha_r}}{\Omega},x\right)$.

By setting $t=\frac{2x-D}{D}$ and employing Chebyshev-Gauss quadrature in the first case, \eqref{cpkjdkj} can be written as 
\begin{align}\label{pcovkj1}
    P_{{\rm cov},kj}(\gamma_{\rm th})&=\sum\limits_{i=0}^{m-1}\left(-\frac{m\gamma_{\rm th}}{\Omega}\right)^{m_i}\frac{1}{m_i!}\frac{3}{D^3}\frac{D}{2}\int\limits_{-1}^{1}\left[\frac{D}{2}(t+1)\right]^{\alpha_rm_i+2} \exp{\left(\tilde{\mathcal{A}}(\gamma_{\rm th},\frac{D}{2}(t+1))\right)}\notag\\
     &~~~\times\bigg[1 +\mathbbm{1}\{m_i>0\}\sum\limits_{n=1}^{m_i}\mathcal{B}\left(m_i,n,\tilde{\mathcal{A}}(\gamma_{\rm th},\frac{D}{2}(t+1))\right)\bigg]dt\notag\\
    &=\sum\limits_{i=0}^{m-1}\left(-\frac{m\gamma_{\rm th}}{\Omega}\right)^{m_i}\frac{1}{m_i!}\frac{3}{2D^2}\frac{\pi}{M_r}\sum\limits_{p=1}^{M_r}b_p^{\alpha_rm_i+2}\exp{\left(\tilde{\mathcal{A}}(\gamma_{\rm th},b_p)\right)} \sqrt{1-t_p^2}\notag\\
     &~~~\times\bigg[1 +\mathbbm{1}\{m_i>0\}\sum\limits_{n=1}^{m_i}\mathcal{B}\left(m_i,n,\tilde{\mathcal{A}}(\gamma_{\rm th},b_p)\right)\bigg],
\end{align}
where $t_p=\cos{\left(\frac{2p-1}{2M_r}\pi\right)}$ and $b_p=\frac{D}{2}(t_p+1)$.

After reorganizing the components in \eqref{pcovkj1}, \eqref{cprf3} can be obtained.
\end{proof}

Moreover, the CDF of $\gamma_{kj}$ can be evaluated as 
\begin{align}
    F_{\gamma_{kj}}(x)=1-P_{{\rm cov},kj}(x).
\end{align}

\subsection{Interference-and-Noise Case}
\subsubsection{Coverage Performance}
If both the interference and the noise are considered, the signal-to-interference-and-noise ratio (SINR) is
\begin{align}
    \gamma_{kj}=\frac{\frac{P_{\rm R}g_{kj}}{\rho d_{kj}^{\alpha_r}}}{\sum\limits_{i=1,i\not=k}^{N}\frac{P_{\rm R}g_{ij}}{\rho d_{ij}^{\alpha_r}}+N_{\rm R}}=\frac{\frac{g_{kj}}{ d_{kj}^{\alpha_r}}}{I+\frac{\rho}{P_{\rm R}}N_{\rm R}}.
\end{align}

\begin{theorem}
Considering the interference and noise, the CP of the RF link in this case can be expressed as 
\begin{small}
\begin{align}\label{cprf4}
    P_{{\rm cov},kj}(\gamma_{\rm th})\approx&\frac{3\pi }{2D^2M_r}\sum\limits_{p=1}^{M_r}\sqrt{1-t_p^2}\sum\limits_{i=0}^{m-1}\left(\frac{m\gamma_{\rm th}}{\Omega}\right)^{m_i}\frac{b_p^{\alpha_rm_i+2}}{m_i!} \exp{\left(-\frac{m\gamma_{\rm th}b_p^{\alpha_r}\rho N_{\rm R}}{\Omega P_{\rm R}}+\tilde{\mathcal{A}}(\gamma_{\rm th},b_p)\right)}\notag\\ 
     &\times\sum\limits_{u=0}^{m_i}\binom{m_i}{u}\left(\frac{\rho N_{\rm R}}{P_{\rm R}}\right)^{m_i-u}(-1)^u\Big[1+\mathbbm{1}\{u>0\} \sum\limits_{n=1}^{u} \mathcal{B}\left(u,n,\tilde{\mathcal{A}}(\gamma_{\rm th},b_p)\right) \Big].
\end{align}
\end{small}

\end{theorem}

\begin{proof}
To include noise, $I$ should be replaced by $I+\frac{\rho}{P}N_{\rm R}$ in \eqref{I2}. According to \cite[Eq. 1.111]{gradshteyn2014table}, we can get the $I_2$ in this case as 
\begin{align}\label{newI3}
    I_2&=\mathbb{E}_{I}\left[\left(I+\frac{\rho}{P_{\rm R}}N_{\rm R}\right)^{m_i}\exp{\left(-sI-s\frac{\rho N_{\rm R}}{P_{\rm R}}\right)}\right]\notag\\
    &=\exp{\left(-s\frac{\rho N_{\rm R}}{P_{\rm R}}\right)}\sum_{u=0}^{m_i}\binom{m_i}{u}\left(\frac{\rho}{P_{\rm R}}N_{\rm R}\right)^{m_i-u} \mathbb{E}_{I}\left[I^u\exp{\left(-sI\right)}\right].
\end{align}

Combining \eqref{pdfdkj}, \eqref{pcovkj}, \eqref{I2}, \eqref{I21}, \eqref{iesi}, and \eqref{newI3}, and then employing Chebyshev-Gauss quadrature in the first case, the CP of the RF link in this case can be obtained as \eqref{cprf4}.
\end{proof}


\section{The e2e Outage Performance}
In the considered system, we assume that the DF relay scheme is implemented at all CHs. Then, the equivalent e2e SNR from S to UAV can be given as $\gamma^{\rm DF}_{eq}=\min\{\gamma_k,\gamma_{kj}\}$.

\begin{corollary}
The e2e OP for S-CH-UAV links can be finally achieved as 
\begin{align}
    P_{{\rm out},e2e}(\gamma_{\rm th})=1- P_{{\rm cov},{\rm{SCH}}_k}(\gamma_{\rm th})P_{{\rm cov},kj}(\gamma_{\rm th}),
\end{align}
where $P_{{\rm cov},{\rm{SCH}}_k}(\gamma_{\rm th})$ is presented as \eqref{pcop} and $P_{{\rm cov},kj}(\gamma_{\rm th})$ is given as \eqref{cpkj5}, \eqref{cprf3}, and \eqref{cprf4} in three cases. 
\end{corollary}
\begin{proof}
     The CDF of the equivalent SNR under DF scheme, $\gamma^{\rm DF}_{eq}$, is given as \cite{ai2019physical}
    \begin{align}
        F_{\gamma_{eq}^{\rm DF}}(x)=1-\left[1-F_{\gamma_k}(x)\right]\left[1-F_{\gamma_{kj}}(x)\right].
    \end{align}
    
    The OP here is defined as the probability that $\gamma^{\rm DF}_{eq}$ falls below a given threshold $\gamma_{\rm th}$, which can be written as
    \begin{align}
        P_{{\rm out},e2e}(\gamma_{\rm th})&=F_{\gamma_{eq}^{\rm DF}}(\gamma_{\rm th})=1-\left[1-F_{\gamma_k}(\gamma_{\rm th})\right]\left[1-F_{\gamma_{kj}}(\gamma_{\rm th})\right]=1- P_{{\rm cov},{\rm{SCH}}_k}(\gamma_{\rm th})P_{{\rm cov},kj}(\gamma_{\rm th}).
    \label{op}
    \end{align}
\end{proof}

\section{Numerical Results}
In this section, numerical results will be provided to study the coverage and outage performance of the considered satellite-UAV systems, as well as to verify the proposed analytical models. In the simulation, we run $1\times10^6$ trials of Monte-Carlo simulations to model the randomness of the positions of the considered CHs and UAVs. Unless otherwise explicitly specified, the main parameters adopted in this section are set in Table \ref{tb1}.

\begin{table}[ht]
\caption{Values of main parameters}
\label{tb1}
\centering
\setlength{\abovecaptionskip}{0pt}
\setlength{\belowcaptionskip}{10pt}
\begin{tabular}{c|c|c|c|c|c|c|c|c}
\hline\hline
Parameter & Value   & Unit & Parameter & Value    & Unit & Parameter & Value    & Unit \\ \hline
$H_{\rm U}$  & 50& km   &$H_{\rm S}$  & 35761   & km   &$R$   & 6376    & km\\ \hline
$\xi_0$ & $\pi/800$ & rad  & $G_{\rm S}$ and $G_{\rm R}$  & 107.85   & dB  & $\lambda$    & 1550     & nm    \\ \hline
 $N_{\rm F}$ and $N_{\rm R}$  & $10^{-10}$     & mW  & $\eta$   & 0.5 &    & $A_0$  & 0.5&\\ \hline
$\lambda_{\rm P}$   & 0.001   &  & $\omega$ & 1.1& & $P_{\rm S}$  & 40& dBm  \\ \hline
$\alpha$ &4 & &$\beta$ & 1.9& & 
$h_l$  & $-0.35$ &dB \\ \hline
$D$  & 1 &km &$D_{\max}$ & 20& km & $D_{\min}$  & $2$ &km\\ \hline
 $\Omega$ & 1& &$m$  & $5$ & &$P_{\rm R}$ & 30 & dBm \\ \hline
 $\alpha_r$ & 2& &$\rho$  &7018  & & &  &  \\ \hline\hline
\end{tabular}
\vspace{-15mm}
\end{table}

\subsection{Performance over S-CH FSO Links}\label{perffso}
In this subsection, we will study the coverage performance over S-CH links. 

\begin{figure}[!htb]
\centering
\setlength{\abovecaptionskip}{0pt}
\setlength{\belowcaptionskip}{10pt}
\subfigure[$P_{{\rm cov},{\rm{SCH}}_k}$ versus $\gamma_{\rm th}$ with $P_{\rm S}=40$ dBm]{
\includegraphics[width=3 in]{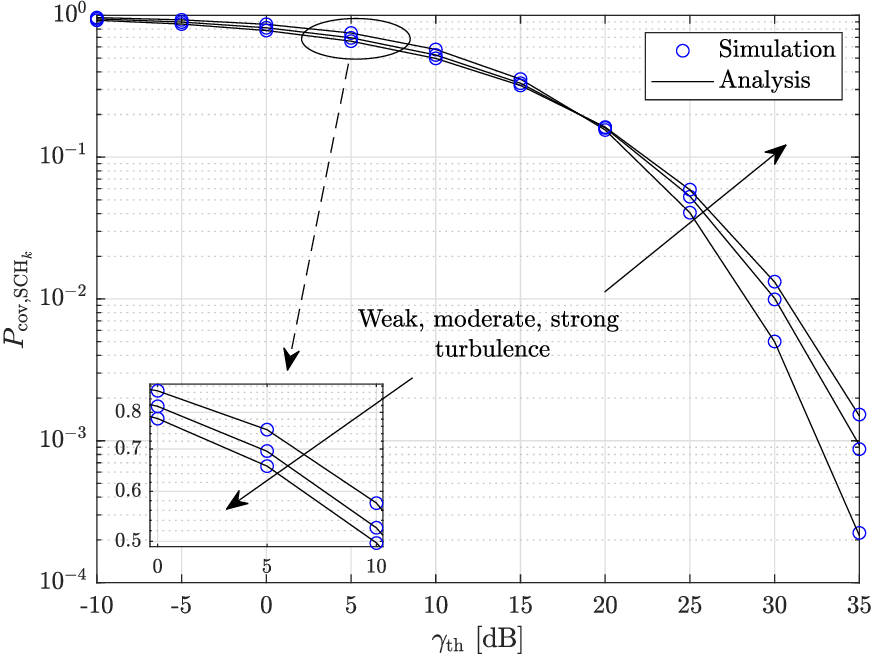}
\label{cpfsoturbu}
}%
\hfill
\subfigure[$1-P_{{\rm cov},{\rm{SCH}}_k}$ versus $P_{\rm S}$ with $\gamma_{\rm th}=30$ dB]{\includegraphics[width=3 in]{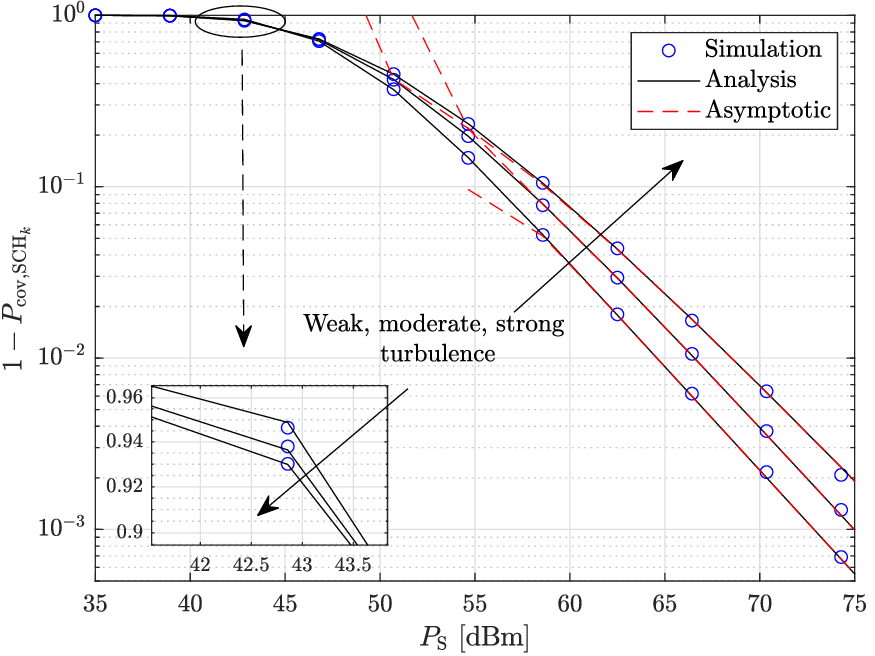}
\label{opfsoturbu1}
}

\subfigure[$1-P_{{\rm cov},{\rm{SCH}}_k}$ versus $P_{\rm S}$ with $\gamma_{\rm th}=5$ dB]{
\includegraphics[width=3 in]{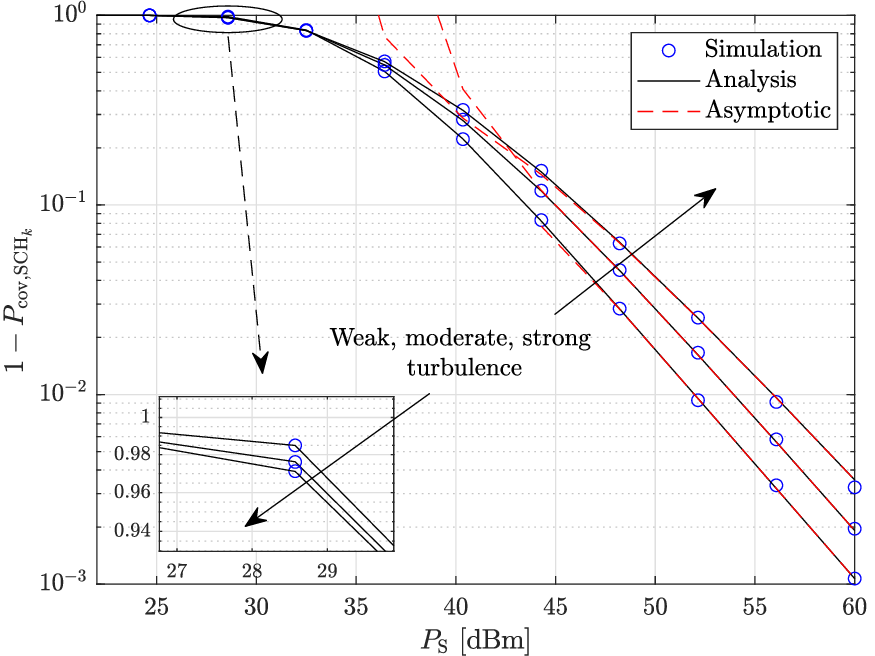}
\label{opfsoturbu2}
}
\centering
\caption{$P_{{\rm cov},{\rm{SCH}}_k}$ and $1-P_{{\rm cov},{\rm{SCH}}_k}$ over S-CH links for various kinds of turbulence (weak turbulence: $\alpha=4.76$, $\beta=3.03$; moderate turbulence: $\alpha=4$, $\beta=1.9$; strong turbulence: $\alpha=4.2$, $\beta=1.4$).}
\label{diffturbu}
\vspace{-10mm}
\end{figure}

In Figs. \ref{diffturbu}, \ref{diffhs}, and \ref{diffpointerror}, the CP is presented for different turbulence, $H_{\rm S}$, and pointing errors, respectively. One can easily see that CP decreases with the increment of $\gamma_{\rm th}$ and coverage performance can be improved while $P_{\rm S}$ increases. The first observation is caused by the fact that a large $\gamma_{\rm th}$ represents a small probability of coverage events. The second observation can be explained by the fact that a large $P_{\rm S}$ generates a large average power of received signals which leads to a large received SNR.

Fig. \ref{cpfsoturbu} shows that the CP with the weak turbulence outperforms the one with the strong turbulence when $\gamma_{\rm th}$ is less than 18 dB. On the contrary, the opposite observation can be found when $\gamma_{\rm th}$ is greater than 18 dB. When $\gamma_{\rm th}$ is large, Fig. \ref{opfsoturbu1} depicts that the weak turbulence leads to a large $1-P_{{\rm cov},{\rm{SCH}}_k}$ (small $P_{{\rm cov},{\rm{SCH}}_k}$) in small $P_{\rm S}$ region, while inverse observation can be obtained in large $P_{\rm S}$ region. When $\gamma_{\rm th}$ is small, $1-P_{{\rm cov},{\rm{SCH}}_k}$ in Fig. \ref{opfsoturbu2} presents the same conclusion compared with that in Fig. \ref{opfsoturbu1}. 

\begin{figure}[!htb]
\centering
\setlength{\abovecaptionskip}{0pt}
\setlength{\belowcaptionskip}{10pt}
\subfigure[$P_{{\rm cov},{\rm{SCH}}_k}$ versus $\gamma_{\rm th}$]{
\includegraphics[width=3 in]{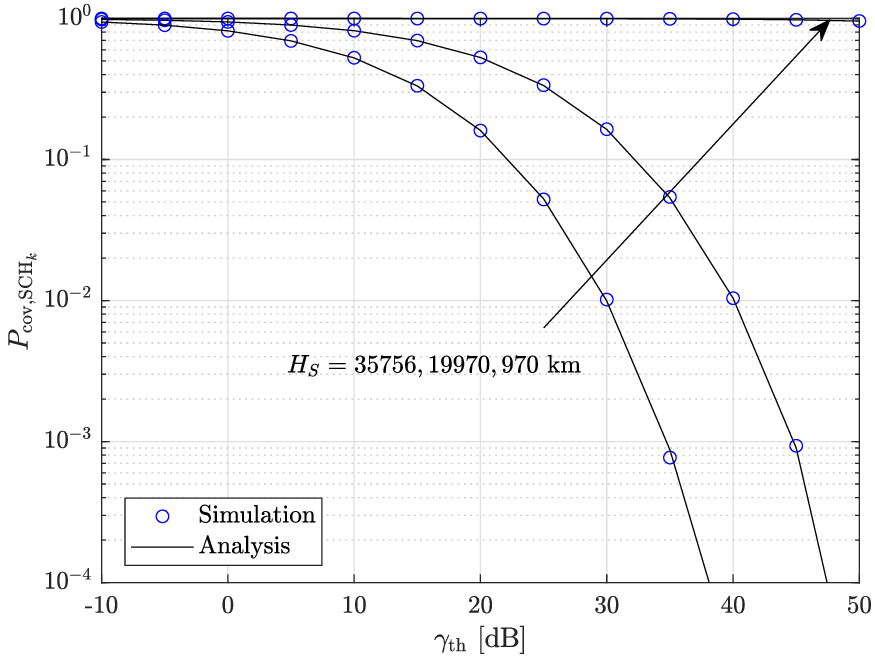}
\label{cpfsohs}
}%
\hfill
\subfigure[$1-P_{{\rm cov},{\rm{SCH}}_k}$ versus $P_{\rm S}$ with $\gamma_{\rm th}=30$ dB]{
\includegraphics[width=3 in]{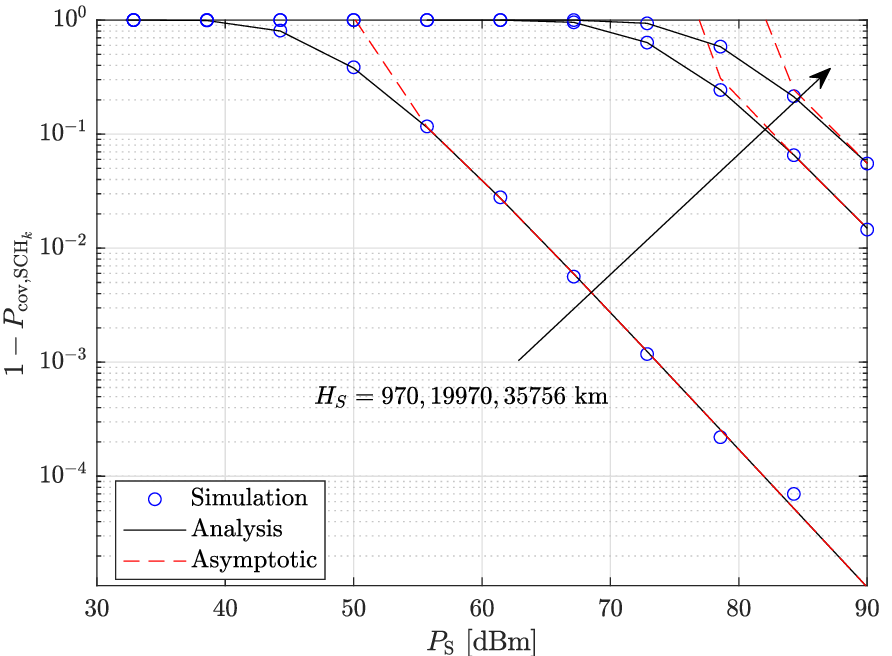}
\label{opfsohs}
}
\centering
\caption{$P_{{\rm cov},{\rm{SCH}}_k}$ and $1-P_{{\rm cov},{\rm{SCH}}_k}$ over S-CH links for various $H_{\rm S}$.}
\label{diffhs}
\vspace{-10mm}
\end{figure}

In Figs. \ref{cpfsohs} and \ref{opfsohs}, we can observe that $H_{\rm S}$ shows a negative effect on the coverage performance. In other words, the CP degrades with an enlarging the orbit height of the satellite which denotes increasing path-loss. Because a large path-loss causes a small received SNR at the CH which results in small CP.

\begin{figure}[!htb]
\centering
\setlength{\abovecaptionskip}{0pt}
\setlength{\belowcaptionskip}{10pt}
\subfigure[$P_{{\rm cov},{\rm{SCH}}_k}$ versus $\gamma_{\rm th}$]{
\includegraphics[width=3 in]{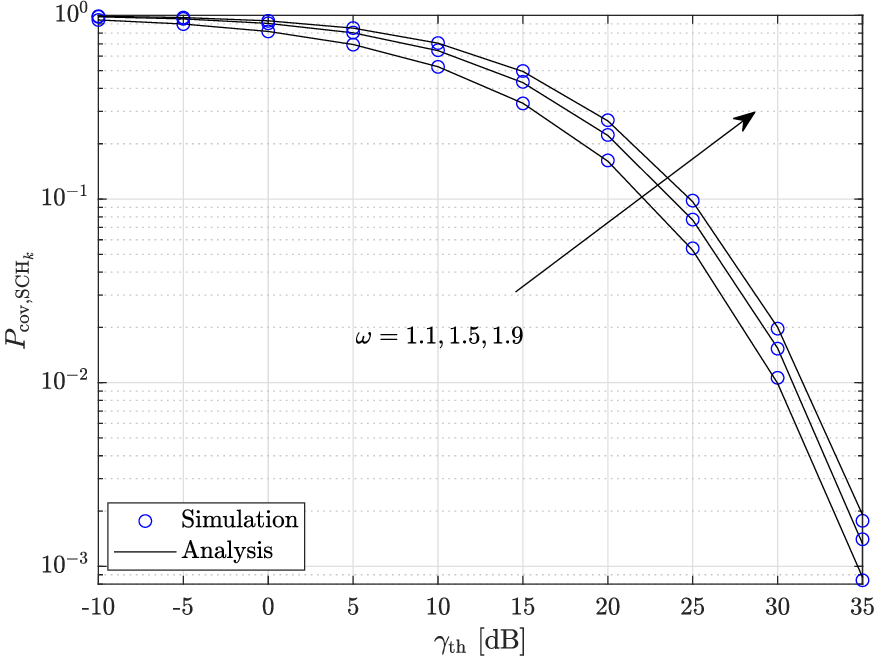}
\label{cpfsopointerror}
}%
\hfill
\subfigure[$1-P_{{\rm cov},{\rm{SCH}}_k}$ versus $P_{\rm S}$ with $\gamma_{\rm th}=20$ dB]{
\includegraphics[width=3 in]{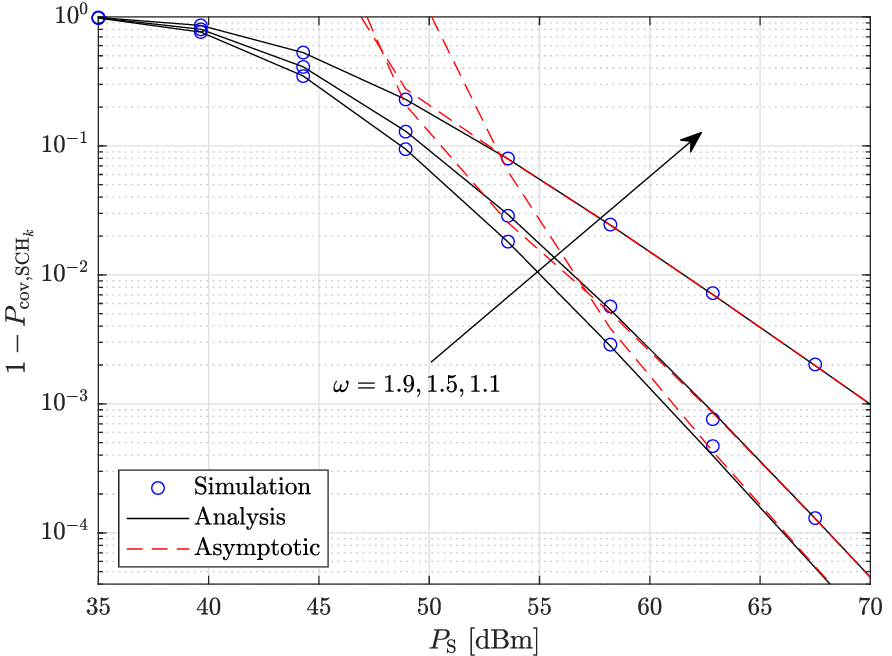}
\label{opfsopointerror}
}
\centering
\caption{$P_{{\rm cov},{\rm{SCH}}_k}$ and $1-P_{{\rm cov},{\rm{SCH}}_k}$ over S-CH FSO links for various pointing errors.}
\label{diffpointerror}
\vspace{-10mm}
\end{figure}

Fig. \ref{diffpointerror} shows that increasing $w$ leads to improved coverage performance. This observation can be explained by the fact that a large $w$ denotes a small pointing error displacement standard deviation with a fixed equivalent beam radius at the receiver, resulting in a large average received power and SNR.

In Figs. \ref{opfsoturbu1}, \ref{opfsoturbu2}, and \ref{opfsohs}, the asymptotic curves in each figure have the same slope in high $P_{\rm S}$ region as they show the same diversity order $\min\{\omega^2,\alpha,\beta\}=1.21$. However, the asymptotic curves with $\omega=1.1$ exhibits a different slope from the others in fig. \ref{opfsopointerror}. Because, $\omega=1.1$ results in $\min\{\omega^2,\alpha,\beta\}=1.21$ and $\omega=1.5$ or 1.9 leads to $\min\{\omega^2,\alpha,\beta\}=1.9$, which cause different diversity orders.

Furthermore, one can clearly see from Figs. \ref{diffturbu}-\ref{diffpointerror} that simulation results agree with the analysis ones very well and asymptotic curves converge to the simulation and analysis ones in high $P_{\rm S}$ region, which verifies the accuracy of our proposed analytical models and the correctness of the derived diversity order.

\subsection{Performance over CH-UAV RF Links}
In this subsection, coverage performance under three cases (interference-free case, interference-dominated case, and interference-and-noise case) will be investigated for various main parameters.

\subsubsection{Interference-Free Case}
\begin{figure}[!htb]
\centering 
 \setlength{\abovecaptionskip}{0pt}
  \setlength{\belowcaptionskip}{10pt}
\begin{minipage}[b]{0.45\textwidth} 
\centering 
\includegraphics[width=3 in]{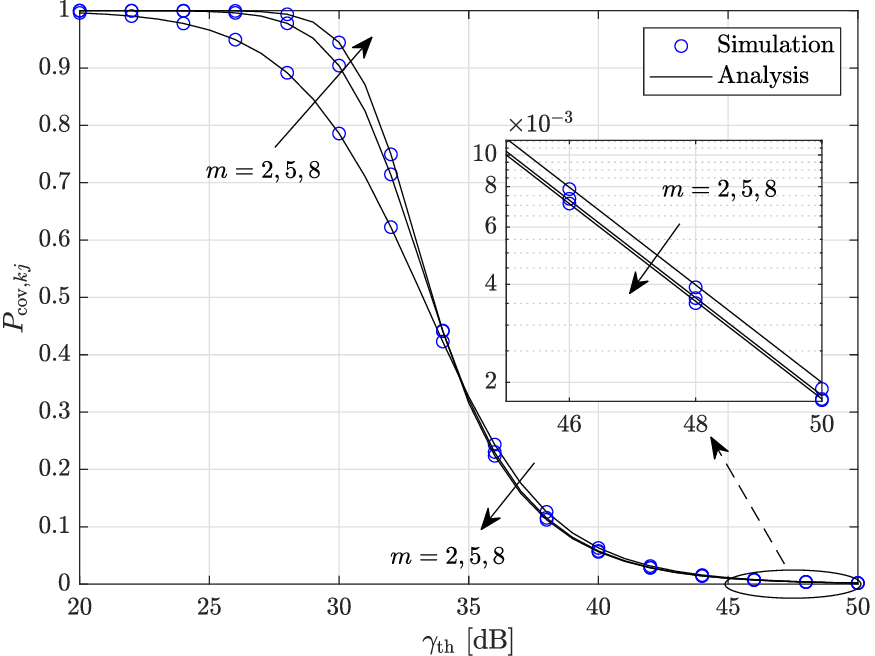}
    \caption{$P_{{\rm cov,}kj}$ versus $\gamma_{\rm th}$ over CH-UAV links for different $m$ in interference-free case}
    \label{snr_diff_m0}
\end{minipage}
\hfill
\begin{minipage}[b]{0.45\textwidth} 
\centering 
\includegraphics[width=3 in]{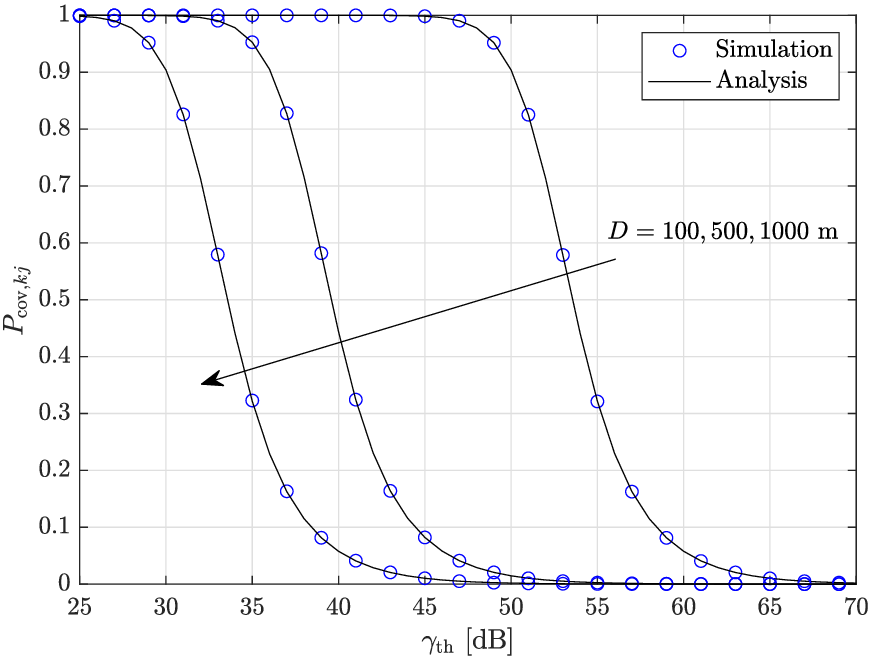}
    \caption{$P_{{\rm cov,}kj}$ versus $\gamma_{\rm th}$ over CH-UAV links for different $D$ in interference-free case}
    \label{snr_diff_D}
\end{minipage}
\vspace{-10mm}
\end{figure}



\begin{figure}[!htb]
\centering
    \setlength{\abovecaptionskip}{0pt}
    \setlength{\belowcaptionskip}{10pt}
    \includegraphics[width=3.1 in]{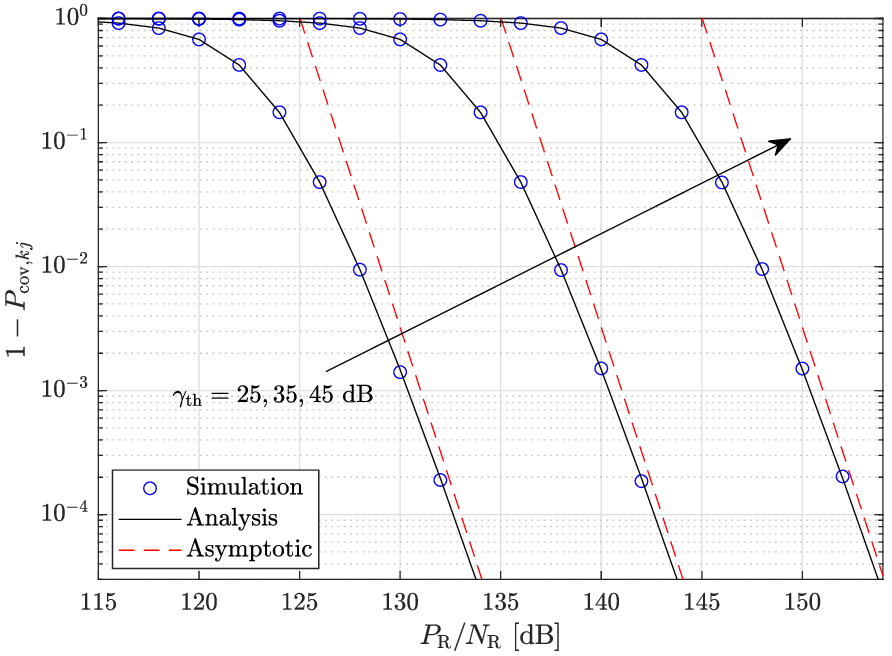}
    \caption{$1-P_{{\rm cov,}kj}$ versus $P_{\rm R}/N_{\rm R}$ over CH-UAV links for different $\gamma_{\rm th}$ in interference-free case }
    \centering
    \label{snr_op_diff_rth}
\end{figure}

Figs. \ref{snr_diff_m0} and \ref{snr_diff_D} presents the coverage performance for various $m$ and $D$ in the case that there is no interference. One can see that the CP with a large $m$ outperforms the one with a small $m$ when $\gamma_{\rm th}$ is less than 34 dB. An opposite observation can be achieved for $\gamma_{\rm th}>34$ dB. Also, it is obvious that a large $D$ leads to a small CP, which denotes a large distributed space for the UAVs, leading to large path-loss.  

Figs. \ref{snr_op_diff_rth} shows the $1-P_{{\rm cov},kj}$ versus $P_{\rm R}/N_{\rm R}$ for various $\gamma_{\rm th}$. It is easy to find out that $1-P_{{\rm cov},kj}$ decreases or $P_{{\rm cov},kj}$ increases when $P_{\rm R}/N_{\rm R}$ increases or $\gamma_{\rm th}$ decreases. These findings can be explained by the reason proposed in the second paragraph in Section \ref{perffso}.

\subsubsection{Interference-Dominated Case}
\begin{figure}[!htb]
\centering 
\setlength{\abovecaptionskip}{0pt}
\setlength{\belowcaptionskip}{10pt}
\begin{minipage}[b]{0.45\textwidth} 
\centering 
\includegraphics[width=3 in]{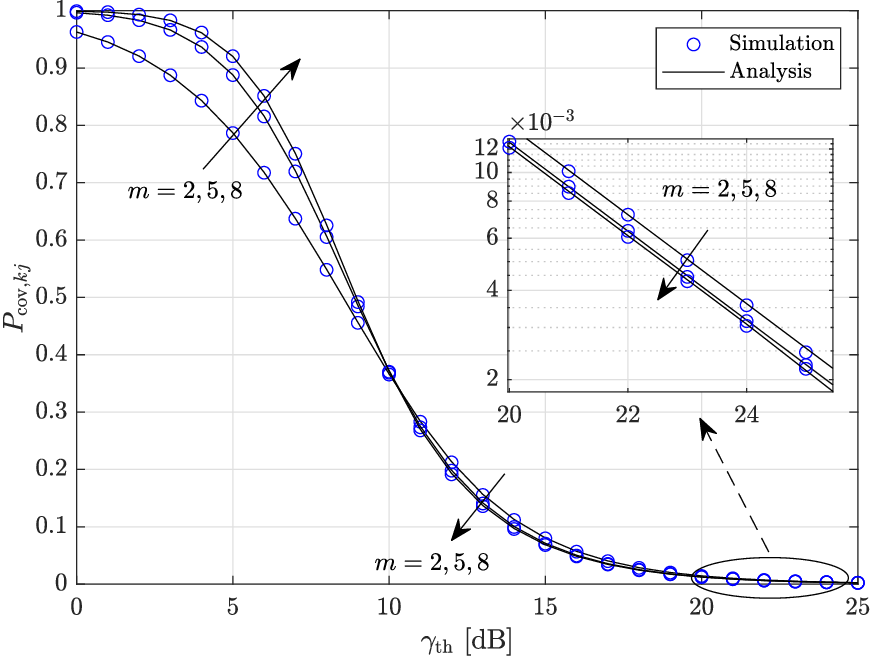}
\caption{$P_{{\rm cov,}kj}$ versus $\gamma_{\rm th}$ over CH-UAV links for different $m$ in interference-dominated case}
\label{sir_diff_m0}
\end{minipage}
\hfill
\begin{minipage}[b]{0.45\textwidth} 
\centering 
\includegraphics[width=3 in]{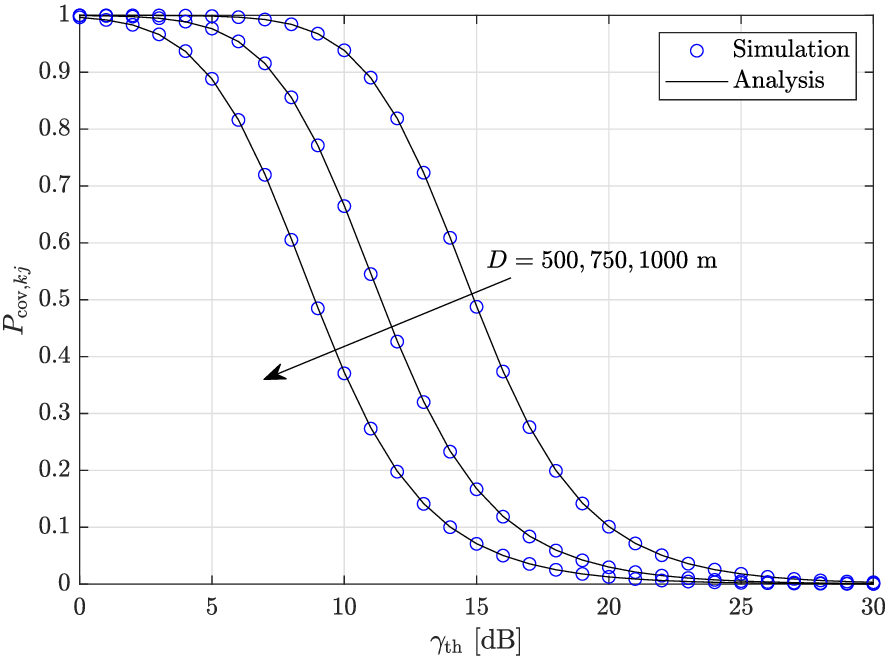}
\caption{$P_{{\rm cov,}kj}$ versus $\gamma_{\rm th}$ over CH-UAV links for different $D$ in interference-dominated case}
\label{sir_diff_D}
\end{minipage}
\vspace{-10mm}
\end{figure}

\begin{figure*}[htbp]
\centering 
\setlength{\abovecaptionskip}{0pt}
\setlength{\belowcaptionskip}{10pt}
\begin{minipage}[b]{0.45\textwidth} 
\centering 
\setlength{\abovecaptionskip}{0pt}
\setlength{\belowcaptionskip}{10pt}
\includegraphics[width=3 in]{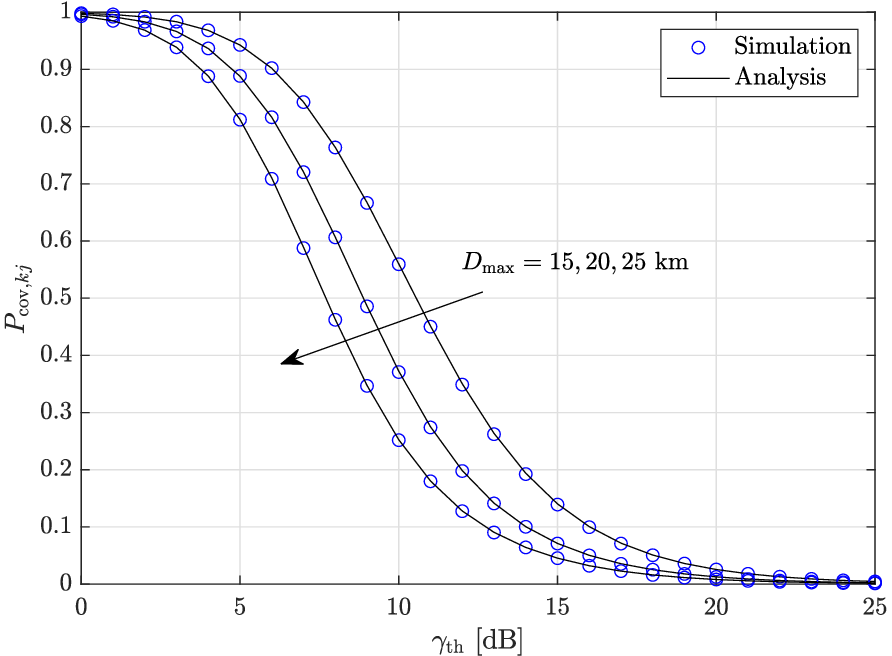}
\caption{$P_{{\rm cov,}kj}$ versus $\gamma_{\rm th}$ over CH-UAV links for different $D_{\max}$ in interference-dominated case}
\label{sir_diff_Dmax}
\end{minipage}
\hfill
\begin{minipage}[b]{0.45\textwidth} 
\centering 
\setlength{\abovecaptionskip}{0pt}
\setlength{\belowcaptionskip}{10pt}
\includegraphics[width=3 in]{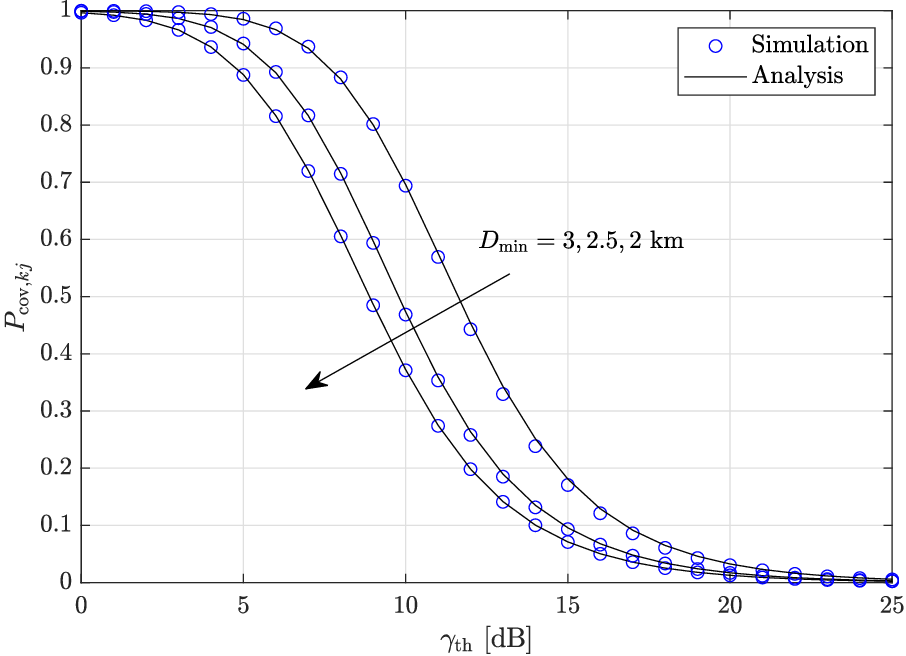}
\caption{$P_{{\rm cov,}kj}$ versus $\gamma_{\rm th}$ over CH-UAV links for different $D_{\min}$ in interference-dominated case}
\label{sir_diff_Dmin}
\end{minipage}
\vspace{-10mm}
\end{figure*}




Figs. \ref{sir_diff_m0} and \ref{sir_diff_D} show the coverage performance over CH-UAV links with various $m$ and $D$ when the interference is dominant. Similar conclusions can be reached here with these achieved in Figs. \ref{snr_diff_m0} and \ref{snr_diff_D}.   

As presented in Figs. \ref{sir_diff_Dmax} and \ref{sir_diff_Dmin}, the influences of $D_{\max}$ and $D_{\min}$ are investigated, respectively. Clearly, $D_{\max}$ shows a negative effect on the CP while $D_{\min}$ exhibits an opposite effect in this case. Because a large $D_{\max}$ or a small $D_{\min}$ leads to more interfering CHs in the considered space, resulting in the degraded coverage performance.

\subsubsection{Interference-and-Noise Case}
\begin{figure}[!htb]
\centering 
\setlength{\abovecaptionskip}{0pt}
\setlength{\belowcaptionskip}{10pt}
\begin{minipage}[b]{0.45\textwidth} 
\centering 
\setlength{\abovecaptionskip}{0pt}
\setlength{\belowcaptionskip}{10pt}
\includegraphics[width=3 in]{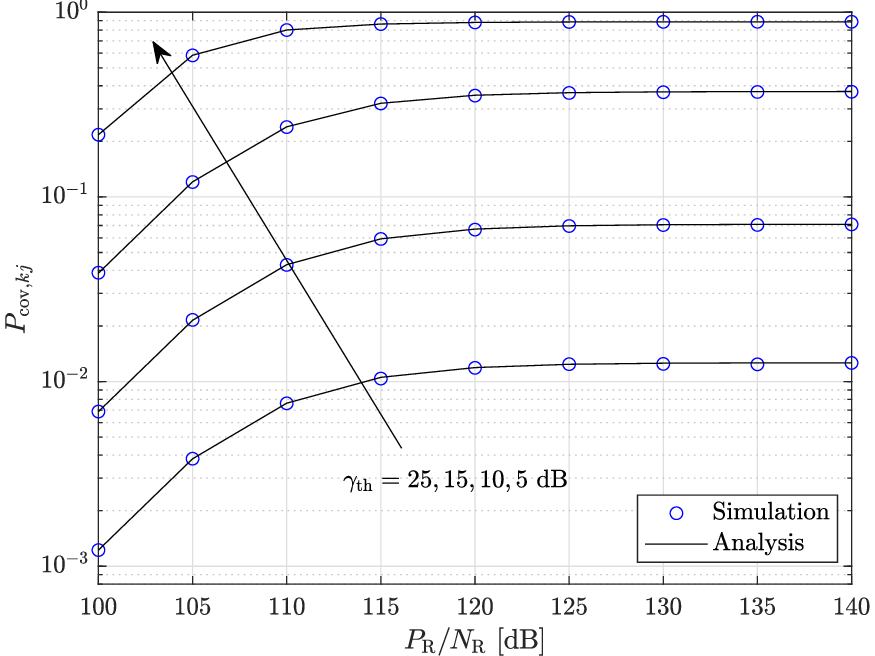}
\caption{$P_{{\rm cov,}kj}$ versus $P_{\rm R}/N_{\rm R}$ over CH-UAV links for various $\gamma_{\rm th}$ in interference-and-noise case }
\label{sinr_diff_th}
\end{minipage}
\hfill
\begin{minipage}[b]{0.45\textwidth} 
\centering 
\setlength{\abovecaptionskip}{0pt}
\setlength{\belowcaptionskip}{10pt}
 \includegraphics[width=3.1 in]{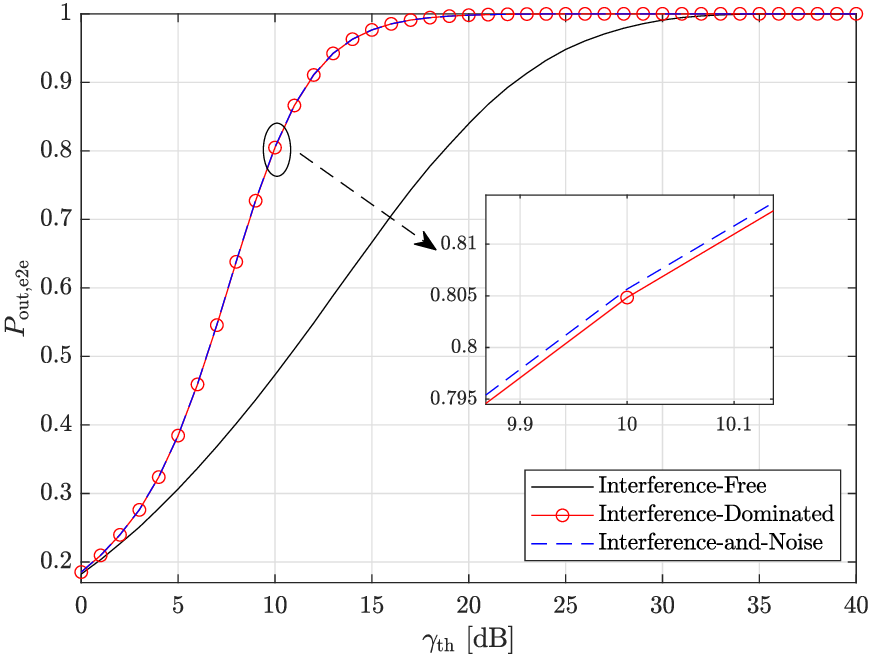}
\caption{$P_{{\rm out,e2e}}$ versus $\gamma_{\rm th}$ over S-CH-UAV links under three cases considered in CH-UAV links}
\label{3cases}
\end{minipage}
\vspace{-10mm}
\end{figure}

Observing the principles about these four parameters achieved in the previous subsection, similar conclusions can be reached in this case, namely, the ones addressing the effects of $m$, $D$, $D_{\max}$, and $D_{\min}$ on CP. Therefore, to more completely understand how the system parameters affect the coverage performance of the target system, the transmit SNR $P_{\rm R}/N_{\rm R}$ of the typical CH will be studied in this subsection, instead of considering the cases of different $m$, $D$, $D_{\max}$, and $D_{\min}$.

Fig. \ref{sinr_diff_th} depicts the CP versus  $P_{\rm R}/N_{\rm R}$ for different $\gamma_{\rm th}$. It is obvious that the CP upgrades with the increment of $P_{\rm R}/N_{\rm R}$ or the decrement of $\gamma_{\rm th}$. The findings can also be explained by the reason proposed in the second paragraph in Section \ref{perffso}.

Furthermore, it can be obviously seen from Figs. \ref{snr_diff_m0}-\ref{sinr_diff_th} that simulation and analysis curves match with each other very well in various cases, which verifies the correctness of the analytical models proposed in Section \ref{rflink} and the approximation shown in Fig. \ref{infspace}.

\subsection{The e2e Outage Performance}
In Fig. \ref{3cases}, the e2e outage performance of the considered system under three cases is investigated through simulation results. One can easily see that the e2e OP in interference-free case outperforms that under the other two cases. We can also observe that the e2e outage performance in interference-dominated case is a little better than that in interference-and-noise case while they are very close to each other. Because the average power of the noise is much smaller than that of the interference, which results in the big gap between the OP in interference-free case and that under the other two cases, as well as the tiny difference between the curves under the latter two cases.

\section{Conclusion}

In this paper, we have studied the coverage and outage performance in a cooperative satellite-UAV communication system with DF relay scheme, while considering the randomness of the positions of CHs and UAVs. Closed-form and approximated expressions for the CP over S-CH FSO links were derived. Moreover, the coverage performance over CH-UAV RF links was analyzed under three cases: interference-free, interference-dominated, and interference-and-noise cases. The analytical expressions for the CP under these three cases were presented, as well as the asymptotic one under the first case. We finally showed the closed-form analytical expression for the e2e OP over S-CH-UAV links. 

Observing from the numerical results, some useful remarks can be reached as follows:

1) The intensity of turbulence exhibits a negative influence on the CP over S-CH FSO link in small $\gamma_{\rm th}$ or large $P_{\rm S}$ region while opposite observation is found in large $\gamma_{\rm th}$ or small $P_{\rm S}$ region.

2) The altitude of the satellite and pointing error show negative influences on the CP over S-CH FSO links. 

3) The fading parameter of Nakagami-$m$ fading, $m$, shows a positive effect on the CP over CH-UAV RF links in small $\gamma_{\rm th}$ region and a negative effect in large $\gamma_{\rm th}$ region.

4) Over CH-UAV RF links, the CH's coverage radius $D$ and sensitivity radius $D_{\max}$ have negative influences on CP while the hard-core radius $D_{\min}$ shows a positive effect. 

5) The diversity orders over S-CH FSO links and CH-UAV links in interference-free case are $\min\{\omega^2,\alpha,\beta\}$ and $m$, respectively.

\section*{Appendix I: Proof of Proposition 1}
\begin{figure}[!htb]
\centering 
\setlength{\abovecaptionskip}{0pt}
\setlength{\belowcaptionskip}{10pt}
\setlength{\abovecaptionskip}{0pt}
\setlength{\belowcaptionskip}{10pt}
\begin{minipage}[b]{0.45\textwidth} 
\centering 
\setlength{\abovecaptionskip}{0pt}
\setlength{\belowcaptionskip}{10pt}
\includegraphics[width=3.1 in]{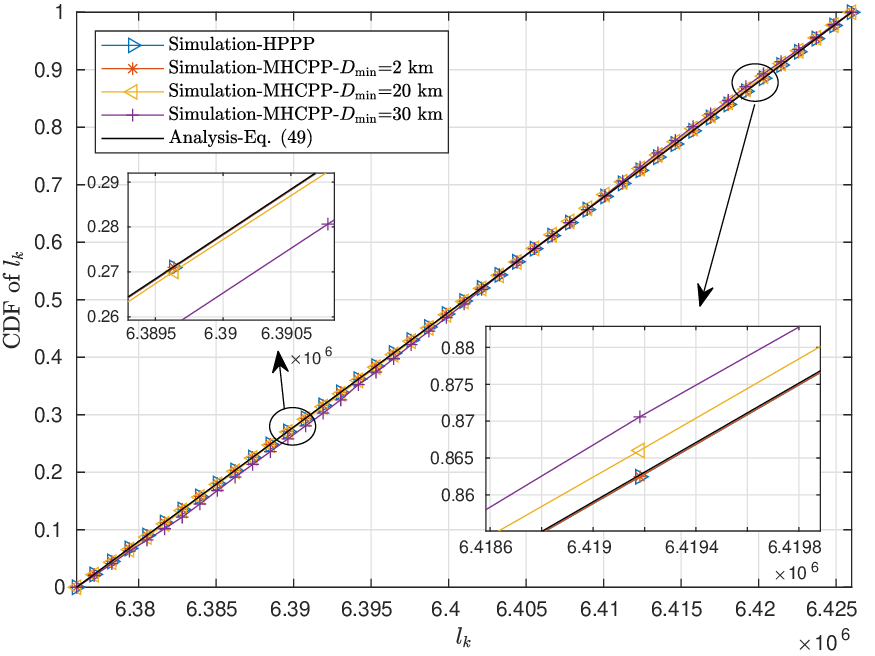}
\caption{CDF of $l_k$}
\label{lk}
\end{minipage}
\hfill
\begin{minipage}[b]{0.45\textwidth} 
\centering 
\setlength{\abovecaptionskip}{0pt}
\setlength{\belowcaptionskip}{10pt}
\includegraphics[width=3.1 in]{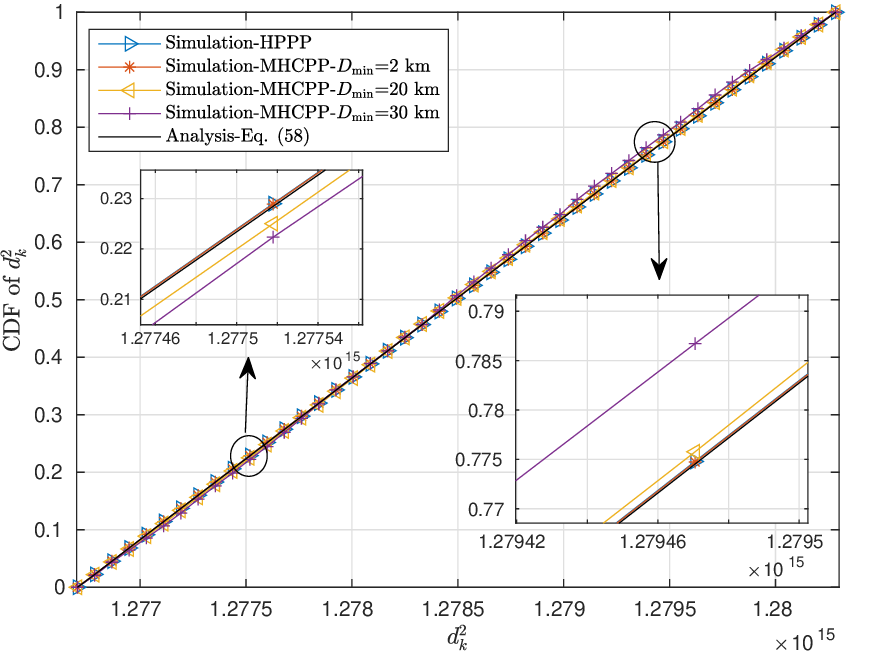}
\caption{CDF of $d^2_k$}
\label{dk}
\end{minipage}
\vspace{-10mm}
\end{figure}



\textcolor{black}{As it is difficult to prove Proposition 1 mathematically, we use Monte-Carlo simulation instead here.} Figs. \ref{lk} and \ref{dk} present simulation and analysis results of the CDFs of $l_k$ and $d^2_k$ with HPPP and MHCPP corresponding to different $D_{\min}$ in \eqref{intensity}. \textcolor{black}{The analysis curves are according to \eqref{cdflkxik} ($y=\xi_0$) and \eqref{cdfdk2} in Appendix II.} The values of other parameters adopted in this simulation are listed in Table \ref{tb1} \textcolor{black}{and $\rm V^{\frac{1}{3}}=46.3$ km.}

\textcolor{black}{
The MHCPP with a specific $D_{\min}$ is thinned from the HPPP by the rule introduced in Section \ref{smodel}. Observing from these two figures, one can see that they present three same rules: 1. the simulation curve of MHCPP gets close to that of HPPP when $D_{\min}$ decreases; 2. the CDF curve of MHCPP with $D_{\min}=2$ km is identical to that that of HPPP; 3. the simulation curve matches the analysis one very well.}

\textcolor{black}{
These three observations verify the correctness of Proposition 1, \eqref{cdflkxik} ($y=\xi_0$) and \eqref{cdfdk2}. }

\section*{Appendix II: Proof of Lemma \ref{lemmapdfdk2}}
The joint CDF of $l_k$ and $\xi_k$ is 
 \begin{align}\label{cdflkxik}
     F_{l_k,\xi_k}(x,y)&=\frac{\int\limits_0^{2\pi}d\theta\int\limits_0^{y} \sin{\xi}d\xi\int\limits_R^{x} r^2dr}{V}=\frac{2\pi (1-\cos{y})(x^3-R^3)}{3V}.
 \end{align}
 
 Then, the joint PDF of $l_k$ and $\xi_k$ can be written as 
 \begin{align}
     f_{l_k,\xi_k}(x,y)=\frac{\partial^2 F_{l_k,\xi_k}(x,y)}{\partial x \partial y}=\frac{2\pi x^2\sin{y}}{V}.
 \end{align}

From Fig. \ref{sh}, the relationships between $l_k$, $\xi_k$, and $d_k^2$ can be represented as
 \begin{align}
     d^2_k={l_k^2+L^2-2l_kL\cos{\xi_k}},
 \end{align}
where $L=H_{\rm S}+H_{\rm U}+R$.

It can be easily seen that 
\begin{align}
d_{\min}=H_{\rm S}\leq d_k\leq \sqrt{(R+H_{\rm U})^2+L^2-2(R+H_{\rm U})L\cos{\xi_0}}=d_{\max}.  \end{align}

To obtain the PDF of $d_k^2$, we first derive the joint PDF of $d_k^2$ and $l_k$.

According to the multivariate change of variables formula, the Jacobian determinant of matrix\\ $\partial(d_k^2,l_k)/\partial(l_k,\xi_k)$ is 
\begin{align}
    \left|\frac{\partial(d_k^2,l_k)}{\partial(l_k,\xi_k)}\right|&=\left|\begin{matrix}
    2l_k-2L\cos{\xi_k}&2l_kL\sin{\xi_k}\\
    1&0    \end{matrix}\right|=2l_kL\sin{\xi_k}.
\end{align}

Then, the joint PDF of $d_k^2$ and $l_k$ can be achieved as
\begin{align}\label{fdklk}
    f_{d^2_k,l_k}(x,y)&=\frac{f_{l_k,\xi_k}(x,y)}{\left|\frac{\partial(d_k^2,l_k)}{\partial(l_k,\xi_k)}\right|}=\frac{\pi y}{VL},
\end{align}
where $R\leq y \leq R+H_{\rm U}$ and $\cos{\xi_0} \leq \cos{\xi_k}=\frac{y^2+L^2-x}{2Ly}\leq1$.

The PDF of $d_k^2$ can be acquired through the integration of \eqref{fdklk} according to $l_k$ as follows
\begin{align}
    f_{d_k^2}(x)&=\int\limits_{\tau_1(x)}^{\tau_2(x)} f_{d_k,l_k}(x,y) dy=\frac{\pi}{2VL}\left[\tau_2^2(x)-\tau_1^2(x)\right].
\end{align}

Observing $R\leq y \leq R+H_{\rm U}$ and $\cos{\xi_0} \leq \cos{\xi_k}=\frac{y^2+L^2-x}{2Ly}\leq1$, $\tau_1$ and $\tau_2$ can be obtained as 
\begin{align}
    \tau_1(x)=\max{\left\{R,L-\sqrt{x} \right\}}
\end{align}
and
\begin{align}
    \tau_2(x)=\min{\left\{R+H_{\rm U},L\cos{\xi_0}-\sqrt{x-L^2\sin^2{\xi_0}}\right\}},
\end{align}
respectively.

Furthermore, the CDF of $d_k^2$ can be achieved as
\begin{align}\label{cdfdk2}
    F_{d_k^2}(y)=\int_{d_{\min}^2}^{y}f_{d_k^2}(x)dx,~d_{\min}^2\leq y\leq d_{\max}^2.
\end{align}

\section*{Appendix III: Proof of Lemma \ref{lemmaesi}}
\begin{figure}[!htb]
\centering
    \setlength{\abovecaptionskip}{0pt}
    \setlength{\belowcaptionskip}{10pt}
    \includegraphics[width=3.1 in]{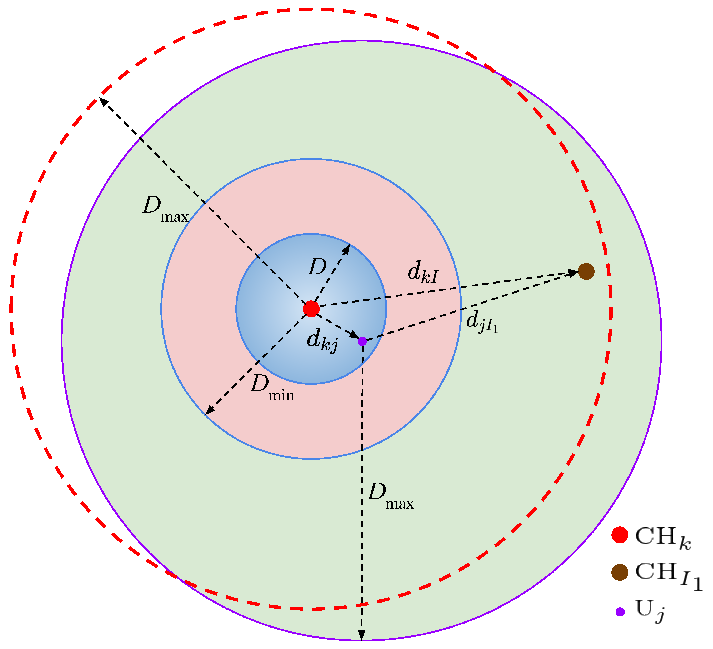}
    \caption{CH-UAV link model}
    \centering
    \label{infspace}
    \vspace{-10mm}
\end{figure}

As shown in Fig. \ref{infspace}, the inner blue sphere is the serving space of CH$_k$. The red and blue spaces are interference-free as it is less than $D_{\min}$ to CH$_k$. These CHs including CH$_{I_1}$, which cause the interference to the typical UAV $U_j$, are located in the green space. $d_{kj}$, $d_{jI_1}$, and $d_{kI_1}$ represents the distances between CH$_k$ and $U_j$, $U_j$ and CH$_{I_1}$, and $U_k$ and CH$_{I_1}$, respectively.

To make the following derivation tractable, we make an approximation that CH$_{I_1}$ is distributed in the sphere with a red dashed outline instead of the green sphere. This approximation is reasonable as these two spheres are mostly overlapped and the shift between them is less than $D$. 

Similar to Lemma \ref{lemmapdfdk2}, we can get the joint PDF of $d_{jI_1}^2$ given $d_{kj}$ as
\begin{align}\label{pdfdji1}
    f_{d_{jI_1}^2|d_{kj}}(x)=\frac{\pi [\tau_4^2(x,d_{kj})-\tau_3^2(x,d_{kj})]}{2d_{kj}V_1},
\end{align}
where $V_1=\frac{4\pi}{3}(D_{\max}^3-D_{\min}^3)$, $\tau_4(x,d_{kj})=\min\{D_{\max},\sqrt{x}+d_{kj}\}$, $\tau_3(x,d_{kj})=\max\{D_{\min},\sqrt{x}-d_{kj}\}$, and $d_{jI_1}^{\min}=(D_{\min}-d_{kj})^2\leq x \leq (D_{\max}+d_{kj})^2=d_{jI_1}^{\max}$.

It is easy to obtain that $V_1$ is the volume of the green space. The number of interfering CHs in this space has the probability ${\bf{Pr}}\{X=N_I\}=\frac{(\lambda_{\rm CH} V_1)^{N_I}}{N_I!}e^{-\lambda_{\rm CH} V_1}$.

For $s>0$, $\mathbb{E}_{I}[e^{-sI}]$ can be calculated as 
\begin{small}
\begin{align}\label{esii6}
    \mathbb{E}_{I}[e^{-sI}]&=\mathbb{E}_{N_I}\left\{\mathbb{E}_{I}\left[\exp{\left(-s\sum_{i=1}^{N_I}g_{jI_i}d_{jI_i}^{-\alpha_r}\right)}\right]\right\}=\mathbb{E}_{N_I}\left\{\prod_{i=1}^{N_I}\mathbb{E}_{g_{jI_i},d_{jI_i}}\left[\exp{\left(-sg_{jI_i}d_{jI_i}^{-\alpha_r}\right)}\right]\right\}\notag\\
    &=\sum_{N_I=0}^{\infty}\frac{(\lambda_{\rm CH} V_1)^{N_I}}{N_I!}\exp{(-\lambda_{\rm CH} V_1)}\left\{\mathbb{E}_{g_{jI_1},d_{jI_1}}\left[\exp{\left(-sg_{jI_1}d_{jI_1}^{-\alpha_r}\right)}\right]\right\}^{N_I}\notag\\
     &=\sum_{N_I=0}^{\infty}\frac{(\lambda_{\rm CH} V_1 I_3)^{N_I}}{N_I!}\exp{(-\lambda_{\rm CH} V_1)}=\exp{(-\lambda_{\rm CH} V_1+\lambda_{\rm CH} V_1I_3))},
\end{align}
\end{small}
where $I_3=\int\limits_{d_{jI_1}^{\min}}^{d_{jI_1}^{\max}}\mathbb{E}_{g_{jI_1}}\left[\exp{\left(-sg_{jI_1}x^{-\frac{\alpha_r}{2}}\right)}\right]f_{d_{jI_1}^{2}|d_{kj}}(x)dx$.

Using the MGF of Nakagami-$m$ function \cite{simon2005digital}, $\mathbb{E}_{g_{jI_1}}\left[\exp{\left(-sg_{jI_1}x^{-\frac{\alpha_r}{2}}\right)}\right]$ can be obtained as
\begin{align}\label{mgfg}
    \mathbb{E}_{g_{jI_1}}\left[\exp{\left(-sg_{jI_1}x^{-\frac{\alpha_r}{2}}\right)}\right]=\left(\frac{\frac{ m}{\Omega s}x^{\frac{\alpha_r}{2}}}{\frac{ m}{\Omega s}x^{\frac{\alpha_r}{2}}+1 }\right)^{m}.
\end{align}

Substituting \eqref{pdfdji1} and \eqref{mgfg} in $I_3$, it deduces
\begin{align}\label{I333}
    I_3&=\left(\frac{m}{\Omega s}\right)^{m}\frac{\pi }{2d_{kj}V_1 }\int\limits_{d_{jI_1}^{\min}}^{d_{jI_1}^{\max}}\frac{x^{\frac{m\alpha_r}{2}}}{\left(\frac{ m}{\Omega s}x^{\frac{\alpha_r}{2}}+1\right)^{m} }\left[\tau_4^2\left(x,d_{kj}\right)-\tau_3^2\left(x,d_{kj}\right)\right]dy.
\end{align}

In the following, we will use $\tau_3$ and $\tau_4$ to respectively represent $\tau_3\left(x,d_{kj}\right)$ and $\tau_4\left(x,d_{kj}\right)$ for convenience and discuss the integral interval in \eqref{I333} in four cases according to the values of $\tau_3$ and $\tau_4$.

Case 1: When $\tau_4=D_{\max}$ and  $\tau_3=D_{\min}$, we can get $(D_{\max}-d_{kj})^2<x<(D_{\min}+d_{kj})^2$. However, as $D_{\max}\gg D_{\min}$, $D_{\max}\gg D>d_{kj}$, and $D_{\max}-D_{\min}\gg 2d_{kj}$, this case does not exist.

Case 2: When $\tau_4=\sqrt{x}+d_{kj}$ and $\tau_3=D_{\min}$, $(D_{\min}-d_{kj})^2<x<(D_{\min}+d_{kj})^2$ and $\tau_4^2-\tau_3^2=x+2d_{kj}\sqrt{x}+d_{kj}^2-D_{\min}^2$ can be obtained.

Case 3: When $\tau_4=\sqrt{x}+d_{kj}$ and $\tau_3=\sqrt{x}-d_{kj}$, it deduces $d_{jI_1}^{\rm{g_1}}=(D_{\min}+d_{kj})^2<x<(D_{\max}-d_{kj})^2=d_{jI_1}^{\rm{g_2}}$ and $\tau_4^2-\tau_3^2=4d_{kj}\sqrt{x}$.

Case 4: When $\tau_4=D_{\max}$ and $\tau_3=\sqrt{x}-d_{kj}$, we can reach $(D_{\max}-d_{kj})^2<x<(D_{\max}+d_{kj})^2$ and $\tau_4^2-\tau_3^2=-x+2d_{kj}\sqrt{x}+D_{\max}^2-d_{kj}^2$.

By using \cite[Eq. 3.194.1]{gradshteyn2014table}, $I_3$ can be represented as
\begin{small}
\begin{align}\label{I6ana}
    I_3&=\left(\frac{m}{\Omega s}\right)^{m}\frac{\pi }{2d_{kj}V_1 }\Bigg\{\int\limits_{d_{jI_1}^{\min}}^{d_{jI_1}^{\rm{g_1}}}\Big[x^{\frac{m\alpha_r}{2}+1}+2d_{kj}x^{\frac{m\alpha_r+1}{2}}+(d_{kj}^2-D_{\min}^2)x^{\frac{m\alpha_r}{2}}\Big]\frac{1}{\left(\frac{ m}{\Omega s}x^{\frac{\alpha_r}{2}}+1\right)^{m} }dx\notag\\
    &~~~+4d_{kj}\int\limits_{d_{jI_1}^{\rm{g_1}}}^{d_{jI_1}^{\rm{g_2}}}\frac{x^{\frac{m\alpha_r+1}{2}}}{\left(\frac{ m}{\Omega s}x^{\frac{\alpha_r}{2}}+1\right)^{m} }dx+\int\limits_{d_{jI_1}^{\rm{g_2}}}^{d_{jI_1}^{\max}}\frac{-x^{\frac{m\alpha_r}{2}+1}+2d_{kj}x^{\frac{m\alpha_r+1}{2}}+(D_{\max}^2-d_{kj}^2)x^{\frac{m\alpha_r}{2}}}{\left(\frac{ m}{\Omega s}x^{\frac{\alpha_r}{2}}+1\right)^{m} }dx\Bigg\}\notag\\
    &=\left(\frac{m}{\Omega s}\right)^{m}\frac{\pi }{2d_{kj}V_1 }\Bigg\{\mathcal{F}\left(\frac{4}{\alpha_r},d_{jI_1}^{\rm{g_1}},d_{jI_1}^{\min}\right)+2d_{kj}\mathcal{F}\left(\frac{3}{\alpha_r},d_{jI_1}^{\rm{g_1}},d_{jI_1}^{\min}\right)+(d_{kj}^2-D_{\min}^2)\mathcal{F}\left(\frac{2}{\alpha_r},d_{jI_1}^{\rm{g_1}},d_{jI_1}^{\min}\right)+4d_{kj}\notag\\
    &\cdot\mathcal{F}\left(\frac{3}{\alpha_r},d_{jI_1}^{\rm{g_2}},d_{jI_1}^{\rm{g_1}}\right)-\mathcal{F}\left(\frac{4}{\alpha_r},d_{jI_1}^{\max},d_{jI_1}^{\rm{g_2}}\right)+2d_{kj}\mathcal{F}\left(\frac{3}{\alpha_r},d_{jI_1}^{\max},d_{jI_1}^{\rm{g_2}}\right)+(D_{\max}^2-d_{kj}^2)\mathcal{F}\left(\frac{2}{\alpha_r},d_{jI_1}^{\max},d_{jI_1}^{\rm{g_2}}\right)\Bigg\},
\end{align}
\end{small}
where 
\begin{align}
    \mathcal{F}(a,b,c)=&\big[\Hypergeometric{2}{1}{m,m+a}{m+a+1}{-\frac{m}{\Omega s}b^{\frac{\alpha_r}{2}}}b^{\frac{(m+a)\alpha_r}{2}}\notag\\
    &-\Hypergeometric{2}{1}{m,m+a}{m+a+1}{-\frac{m}{\Omega s}c^{\frac{\alpha_r}{2}}} c^{\frac{(m+a)\alpha_r}{2}}\big]\frac{2}{(m+a)\alpha_r}\notag
\end{align}
and $\Hypergeometric{2}{1}{\cdot,\cdot}{\cdot}{\cdot}$ denotes Gauss hypergeometric function.

Finally, \eqref{esi} can be obtained by substituting \eqref{I6ana} in \eqref{esii6}.

\section*{Appendix IV: Proof of Lemma \ref{lemmaans}}
According to Leibnitz' rule, $\mathcal{A}^{(n)}(s,d_{kj})$ ($n>0$) can be expressed as
\begin{align}\label{Asnf1f2}
    \mathcal{A}^{(n)}(s,d_{kj})=\frac{\lambda_{\rm CH} \pi}{2d_{kj}}\left(\frac{m}{\Omega}\right)^{m}\sum\limits_{l=0}^{n}\binom{n}{l}f_1^{(n-l)}(s)f_2^{(l)}(s,d_{kj}),
\end{align}
where $f_1(s)=s^{-m}$ and $f_2(s,d_{kj})$ is presented as
\begin{align}\label{f2}
    f_2(s,d_{kj})=&\mathcal{F}\left(\frac{4}{\alpha_r},d_{jI_1}^{\rm{g_1}},d_{jI_1}^{\min}\right)+(d_{kj}^2-D_{\min}^2)\mathcal{F}\left(\frac{2}{\alpha_r},d_{jI_1}^{\rm{g_1}},d_{jI_1}^{\min}\right)\notag\\
    &+2d_{kj}\mathcal{F}\left(\frac{3}{\alpha_r},d_{jI_1}^{\rm{g_1}},d_{jI_1}^{\min}\right)+4d_{kj}\mathcal{F}\left(\frac{3}{\alpha_r},d_{jI_1}^{\rm{g_2}},d_{jI_1}^{\rm{g_1}}\right)-\mathcal{F}\left(\frac{4}{\alpha_r},d_{jI_1}^{\max},d_{jI_1}^{\rm{g_2}}\right)\notag\\
    &+2d_{kj}\mathcal{F}\left(\frac{3}{\alpha_r},d_{jI_1}^{\max},d_{jI_1}^{\rm{g_2}}\right)+(D_{\max}^2-d_{kj}^2)\mathcal{F}\left(\frac{2}{\alpha_r},d_{jI_1}^{\max},d_{jI_1}^{\rm{g_2}}\right).
\end{align}

It is easy to get the $(n-l)$th derivative of $f_1(s)$ as
\begin{align}\label{f1nl}
    f_1^{(n-l)}(s)=(-1)^{n-l}(m)_{n-l}s^{-m-n+l},
\end{align}
where $(m)_{n-l}=\prod\limits_{k=0}^{n-l-1}(m-k)$ is the rising Pochhammer symbol.

To get the $l$th derivative of $f_2(s,d_{kj})$, we should first calculate the $l$th derivative of $\Delta(s,a,b)=\Hypergeometric{2}{1}{m,m+a}{m+a+1}{-\frac{m}{\Omega s}b^{\frac{\alpha_r}{2}}}$.

When $l=0$, $\Delta^{(0)}(s,a,b)=\Delta(s,a,b)$ can be acquired.

When $l>0$, employing the Faà di Bruno's formula, $\Delta^{(l)}(s)$ can be obtained as
\begin{align}\label{deltal}
    \Delta^{(l)}(s,a,b)&=\sum_{q=1}^{l} \frac{(m+a)(m)_{q}}{m+a+q} \Hypergeometric{2}{1}{m+q,m+a+q}{m+a+q+1}{-\frac{mb^{\frac{\alpha_r}{2}}}{\Omega s}}\notag\\
    &~~~~\times B_{l,q}\left(f_3^{(1)}(s,b),...,f_3^{(l-q+1)}(s,b)\right),
\end{align}
where $f_3(s,b)=-\frac{mb^{\frac{\alpha_r}{2}}}{\Omega }s^{-1}$, and $f_3^{(q)}(s,b)=(-1)^{q+1}q!\frac{mb^{\frac{\alpha_r}{2}}}{\Omega }s^{-1-q}$.

Furthermore, the $l$th derivative of $\mathcal{F}(a,b,c)$ can be achieved as
\begin{align}\label{fl}
    \mathcal{F}^{(l)}(a,b,c)=&\frac{2}{(m+a)\alpha_r}\Big[b^{\frac{(m+a)\alpha_r}{2}}\Delta^{(l)}(s,a,b)-c^{\frac{(m+a)\alpha_r}{2}}\Delta^{(l)}(s,a,c)\Big].
\end{align}

Combining \eqref{f2}, \eqref{deltal}, and \eqref{fl},  we can get the $l$th derivative of $f_2(s)$, $f_2^{(l)}(s)$ . Then, \eqref{nAq} can be achieved by substituting \eqref{f1nl} and $f_2^{(l)}(s)$ in \eqref{Asnf1f2}.

\bibliography{citation}

\begin{thebibliography}{10}
\providecommand{\url}[1]{#1}
\csname url@samestyle\endcsname
\providecommand{\newblock}{\relax}
\providecommand{\bibinfo}[2]{#2}
\providecommand{\BIBentrySTDinterwordspacing}{\spaceskip=0pt\relax}
\providecommand{\BIBentryALTinterwordstretchfactor}{4}
\providecommand{\BIBentryALTinterwordspacing}{\spaceskip=\fontdimen2\font plus
\BIBentryALTinterwordstretchfactor\fontdimen3\font minus
  \fontdimen4\font\relax}
\providecommand{\BIBforeignlanguage}[2]{{%
\expandafter\ifx\csname l@#1\endcsname\relax
\typeout{** WARNING: IEEEtran.bst: No hyphenation pattern has been}%
\typeout{** loaded for the language `#1'. Using the pattern for}%
\typeout{** the default language instead.}%
\else
\language=\csname l@#1\endcsname
\fi
#2}}
\providecommand{\BIBdecl}{\relax}
\BIBdecl

\bibitem{pan2020performance}
G.~{Pan}, J.~{Ye}, Y.~{Zhang}, and M.-S. {Alouini}, ``Performance analysis and
  optimization of cooperative satellite-aerial-terrestrial systems,''
  \emph{IEEE Trans. Wireless Commun.}, vol.~19, no.~10, pp. 6693--6707, 2020.

\bibitem{zedini2020performance}
E.~{Zedini}, A.~{Kammoun}, and M.~S. {Alouini}, ``Performance of multibeam very
  high throughput satellite systems based on {FSO} feeder links with {HPA}
  nonlinearity,'' \emph{IEEE Trans. Wireless Commun.}, vol.~19, no.~9, pp.
  5908--5923, 2020.

\bibitem{pan2020harq}
G.~{Pan}, J.~{Ye}, Y.~{Tian}, and M.-S. {Alouini}, ``On {HARQ} schemes in
  satellite-terrestrial transmissions,'' \emph{IEEE Trans. Wireless Commun.},
  pp. 1--1, 2020.

\bibitem{Kawamoto2020flex}
Y.~{Kawamoto}, T.~{Kamei}, M.~{Takahashi}, N.~{Kato}, A.~{Miura}, and
  M.~{Toyoshima}, ``Flexible resource allocation with inter-beam interference
  in satellite communication systems with a digital channelizer,'' \emph{IEEE
  Trans. Wireless Commun.}, vol.~19, no.~5, pp. 2934--2945, 2020.

\bibitem{Christopoulos2015muliti}
D.~{Christopoulos}, S.~{Chatzinotas}, and B.~{Ottersten}, ``Multicast
  multigroup precoding and user scheduling for frame-based satellite
  communications,'' \emph{IEEE Trans. Wireless Commun.}, vol.~14, no.~9, pp.
  4695--4707, 2015.

\bibitem{zhang2020secrecy}
Y.~Zhang, J.~Ye, G.~Pan, and M.-S. Alouini, ``Secrecy outage analysis for
  satellite-terrestrial downlink transmissions,'' \emph{IEEE Wireless Commun.
  Lett.}, vol.~9, no.~10, pp. 1643--1647, 2020.

\bibitem{illi2020phy}
E.~{Illi}, F.~{El Bouanani}, F.~{Ayoub}, and M.-S. {Alouini}, ``A {PHY} layer
  security analysis of a hybrid high throughput satellite with an optical
  feeder link,'' \emph{IEEE Open J. Commun. Soc.}, vol.~1, pp. 713--731, 2020.

\bibitem{zolanvari2020potential}
M.~{Zolanvari}, R.~{Jain}, and T.~{Salman}, ``Potential data link candidates
  for civilian unmanned aircraft systems: A survey,'' \emph{IEEE Commun. Surv.
  Tut.}, vol.~22, no.~1, pp. 292--319, 2020.

\bibitem{cai2020joint}
Y.~{Cai}, Z.~{Wei}, R.~{Li}, D.~W.~K. {Ng}, and J.~{Yuan}, ``Joint trajectory
  and resource allocation design for energy-efficient secure {UAV}
  communication systems,'' \emph{IEEE Trans. Commun.}, vol.~68, no.~7, pp.
  4536--4553, 2020.

\bibitem{lei2020safeguarding}
H.~{Lei}, D.~{Wang}, K.~{Park}, I.~S. {Ansari}, J.~{Jiang}, G.~{Pan}, and M.-S.
  {Alouini}, ``Safeguarding {UAV IoT} communication systems against randomly
  located eavesdroppers,'' \emph{IEEE Internet of Things J.}, vol.~7, no.~2,
  pp. 1230--1244, 2020.

\bibitem{savkin2020securing}
A.~V. {Savkin}, H.~{Huang}, and W.~{Ni}, ``Securing {UAV} communication in the
  presence of stationary or mobile eavesdroppers via online {3D} trajectory
  planning,'' \emph{IEEE Wireless Commun. Lett.}, vol.~9, no.~8, pp.
  1211--1215, 2020.

\bibitem{pan2020secrecy}
G.~{Pan}, H.~{Lei}, J.~{An}, S.~{Zhang}, and M.-S. {Alouini}, ``On the secrecy
  of {UAV} systems with linear trajectory,'' \emph{IEEE Trans. Wireless
  Commun.}, vol.~19, no.~10, pp. 6277--6288, 2020.

\bibitem{li2020unified}
X.~{Li}, Q.~{Wang}, H.~{Peng}, H.~{Zhang}, D.~{Do}, K.~M. {Rabie}, R.~{Kharel},
  and C.~C. {Cavalcante}, ``A unified framework for {HS-UAV NOMA} networks:
  Performance analysis and location optimization,'' \emph{IEEE Access}, vol.~8,
  pp. 13\,329--13\,340, 2020.

\bibitem{sharma2020outage}
P.~K. {Sharma}, D.~{Deepthi}, and D.~I. {Kim}, ``Outage probability of 3-{D}
  mobile {UAV} relaying for hybrid satellite-terrestrial networks,'' \emph{IEEE
  Commun. Lett.}, vol.~24, no.~2, pp. 418--422, 2020.

\bibitem{huang2020energy}
Q.~{Huang}, M.~{Lin}, J.~{Wang}, T.~A. {Tsiftsis}, and J.~{Wang}, ``Energy
  efficient beamforming schemes for satellite-aerial-terrestrial networks,''
  \emph{IEEE Trans. Commun.}, vol.~68, no.~6, pp. 3863--3875, 2020.

\bibitem{kong2020multiuser}
H.~{Kong}, M.~{Lin}, W.~P. {Zhu}, H.~{Amindavar}, and M.~S. {Alouini},
  ``Multiuser scheduling for asymmetric {FSO/RF} links in
  satellite-{UAV}-terrestrial networks,'' \emph{IEEE Wireless Commun. Lett.},
  vol.~9, no.~8, pp. 1235--1239, 2020.

\bibitem{dai2020uav}
H.~{Dai}, H.~{Bian}, C.~{Li}, and B.~{Wang}, ``{UAV}-aided wireless
  communication design with energy constraint in space-air-ground integrated
  green {IoT} networks,'' \emph{IEEE Access}, vol.~8, pp. 86\,251--86\,261,
  2020.

\bibitem{zhou2019beam}
P.~{Zhou}, X.~{Fang}, Y.~{Fang}, R.~{He}, Y.~{Long}, and G.~{Huang}, ``Beam
  management and self-healing for {mmWave} {UAV} mesh networks,'' \emph{IEEE
  Trans. Veh. Technol.}, vol.~68, no.~2, pp. 1718--1732, 2019.

\bibitem{haenggi2012stochastic}
M.~Haenggi, \emph{{Stochastic Geometry for Wireless Networks}}.\hskip 1em plus
  0.5em minus 0.4em\relax Cambridge University Press, 2012.

\bibitem{hunter2008trans}
A.~M. {Hunter}, J.~G. {Andrews}, and S.~{Weber}, ``Transmission capacity of ad
  hoc networks with spatial diversity,'' \emph{IEEE Trans. Wireless Commun.},
  vol.~7, no.~12, pp. 5058--5071, 2008.

\bibitem{he2016modeling}
H.~{He}, J.~{Xue}, T.~{Ratnarajah}, F.~A. {Khan}, and C.~B. {Papadias},
  ``Modeling and analysis of cloud radio access networks using {Matérn}
  hard-core point processes,'' \emph{IEEE Trans. Wireless Commun.}, vol.~15,
  no.~6, pp. 4074--4087, 2016.

\bibitem{omri2018distance}
A.~{Omri} and M.~O. {Hasna}, ``A distance-based mode selection scheme for
  {D2D}-enabled networks with mobility,'' \emph{IEEE Trans. Wireless Commun.},
  vol.~17, no.~7, pp. 4326--4340, 2018.

\bibitem{li2020air}
Y.~{Li}, N.~I. {Miridakis}, T.~A. {Tsiftsis}, G.~{Yang}, and M.~{Xia},
  ``Air-to-air communications beyond {5G}: A novel {3D CoMP} transmission
  scheme,'' \emph{IEEE Trans. Wireless Commun.}, pp. 1--1, 2020.

\bibitem{bach2018model}
S.~{Bachtobji}, A.~{Omri}, R.~{Bouallegue}, and K.~{Raoof}, ``Modelling and
  performance analysis of {mmWaves} and radio-frequency based {3D}
  heterogeneous networks,'' \emph{IET Commun.}, vol.~12, no.~3, pp. 290--296,
  2018.

\bibitem{omri2016model}
A.~{Omri} and M.~O. {Hasna}, ``Modeling and performance analysis of {D2D}
  communications with interference management in {3-D} {HetNets},'' in
  \emph{2016 IEEE Global Commun. Conf. (GLOBECOM)}, 2016, pp. 1--7.

\bibitem{chiu2013stochastic}
S.~N. Chiu, D.~Stoyan, W.~S. Kendall, and J.~Mecke, \emph{{Stochastic Geometry
  and Its Applications}}.\hskip 1em plus 0.5em minus 0.4em\relax John Wiley \&
  Sons, 2013.

\bibitem{ansari2016performance}
I.~S. {Ansari}, F.~{Yilmaz}, and M.-S. {Alouini}, ``Performance analysis of
  free-space optical links over {Málaga} ($\mathcal{M} $) turbulence channels
  with pointing errors,'' \emph{IEEE Trans. Wireless Commun.}, vol.~15, no.~1,
  pp. 91--102, 2016.

\bibitem{goldsmith2005wireless}
A.~Goldsmith, \emph{{Wireless Communications}}.\hskip 1em plus 0.5em minus
  0.4em\relax Cambridge university press, 2005.

\bibitem{pan20173d}
G.~{Pan}, H.~{Lei}, Z.~{Ding}, and Q.~{Ni}, ``On 3-{D} hybrid {VLC-RF} systems
  with light energy harvesting and {OMA} scheme over {RF} links,'' in
  \emph{2017 IEEE Global Commun. Conf. (GLOBECOM)}, 2017, pp. 1--6.

\bibitem{gradshteyn2014table}
I.~S. Gradshteyn and I.~M. Ryzhik, \emph{{Table of Integrals, Series, and
  Products}}.\hskip 1em plus 0.5em minus 0.4em\relax Academic press, 2014.

\bibitem{ansari2015performance}
I.~S. {Ansari}, F.~{Yilmaz}, and M.-S. {Alouini}, ``Performance analysis of
  {FSO} links over unified {Gamma-Gamma} turbulence channels,'' in \emph{2015
  IEEE 81st Veh. Technol. Conf. (VTC Spring)}, 2015, pp. 1--5.

\bibitem{lei2013outage}
X.~{Lei}, L.~{Fan}, D.~S. {Michalopoulos}, P.~{Fan}, and R.~Q. {Hu}, ``Outage
  probability of {TDBC} protocol in multiuser two-way relay systems with
  {Nakagami-$m$} fading,'' \emph{IEEE Commun. Lett.}, vol.~17, no.~3, pp.
  487--490, 2013.

\bibitem{ai2019physical}
Y.~{Ai}, A.~{Mathur}, M.~{Cheffena}, M.~R. {Bhatnagar}, and H.~{Lei},
  ``Physical layer security of hybrid satellite-{FSO} cooperative systems,''
  \emph{IEEE Photon. J.}, vol.~11, no.~1, pp. 1--14, 2019.

\bibitem{simon2005digital}
M.~K. Simon and M.-S. Alouini, \emph{{Digital Communication over Fading
  Channels}}.\hskip 1em plus 0.5em minus 0.4em\relax John Wiley \& Sons, 2005,
  vol.~95.

\end{thebibliography}
\end{document}